\documentclass[10pt]{article}
\usepackage{times}
\usepackage{helvet}
\usepackage{courier}
\usepackage{paralist}
\usepackage{latexsym}
\usepackage{url}
\usepackage[all]{xy}
\usepackage{amsmath}
\usepackage{amssymb}
\usepackage{comment}
\usepackage{xcolor}
\usepackage{times}
\usepackage{graphicx}
\usepackage{savesym}
\savesymbol{AND}
\usepackage{xspace}
\usepackage{tikz}
\usepackage{pgfplots}
\usepackage{pgf}
\usepackage{natbib}
\usepackage{algorithm}
\usepackage{algorithmic}
\usepackage{xspace}
\usepackage{comment}

\usepackage{ssltr}
\usepackage{macros}
\usetikzlibrary{intersections}
\usetikzlibrary{arrows,calc,fit,patterns,plotmarks,shapes.geometric,shapes.misc,shapes.symbols,   shapes.arrows,   shapes.callouts,   shapes.multipart,   shapes.gates.logic.US,   shapes.gates.logic.IEC,   er,   automata,   backgrounds,   chains,   topaths,   trees,   petri,   mindmap,   matrix,   calendar,   folding, fadings,   through,   positioning,   scopes,   decorations.fractals,   decorations.shapes,   decorations.text,   decorations.pathmorphing,   decorations.pathreplacing,   decorations.footprints,   decorations.markings, shadows,circuits}

\newcommand{\nn}{\nonumber}
\newcommand{\minp}{(\min,+)}
\newcommand{\R}{\mathbf{R}}
\newcommand{\Rm}{\mathbf{R}_{\min}}
\newcommand{\ra}{\rightarrow}
\newcommand{\om}{\otimes}
\newcommand{\op}{\oplus}
\newcommand{\RA}{\Rightarrow}
\newcommand{\LA}{\Leftarrow}

\newenvironment{proof}{{\bf Proof:} }{}
\newtheorem{theorem}{Theorem}
\newtheorem{lemma}[theorem]{Lemma}
\newtheorem{assumption}{Assumption}
\newtheorem{definition}[theorem]{Definition}

\newtheorem{corollary}{Corollary}

\newtheorem{example}{Example}

\def\R{\mathrm{R}}

\newcounter{subequation}[equation]

\def\mathdisplay#1{%
  \ifmmode \@badmath
  \else
    $$\def\@currenvir{#1}%
    \let\dspbrk@context\z@
    \let\tag\tag@in@display \SK@equationtrue %\let\label\label@in@display
    \global\let\df@label\@empty \global\let\df@tag\@empty
    \global\tag@false
    \let\mathdisplay@push\mathdisplay@@push
    \let\mathdisplay@pop\mathdisplay@@pop
    \if@fleqn
      \edef\restore@hfuzz{\hfuzz\the\hfuzz\relax}%
      \hfuzz\maxdimen
      \setbox\z@\hbox to\displaywidth\bgroup
        \let\split@warning\relax \restore@hfuzz
        \everymath\@emptytoks \m@th $\displaystyle
    \fi
}

\newcounter{algostep}

\newcounter{acalgorithm}

\trinfo{2013}{11}{30}
\title{Approximate Dynamic Programming based on Projection onto the $\minp$ subsemimodule}
\author{Chandrashekar L$^\dag$ \and Shalabh Bhatnagar$^\$$}
\affiliation{$^\dag$Dept Of CSA, IISc {\tt chandrul@csa.iisc.ernet.in}, $^\$$Dept Of CSA, IISc {\tt shalabh@csa.iisc.ernet.in}}
\everymath{\displaystyle}
\begin{document}
%\maketrcover
\maketitle
\begin{abstract}
We develop a new Approximate Dynamic Programming (ADP) method for infinite horizon discounted reward Markov Decision Processes (MDP) based on projection onto a subsemimodule. We approximate the value function in terms of a $\minp$ linear combination of a set of basis functions whose $\minp$ linear span constitutes a subsemimodule. The projection operator is closely related to the $\emph{Fenchel}$ transform. Our approximate solution obeys the $\minp$ Projected Bellman Equation (MPPBE) which is different from the conventional Projected Bellman Equation (PBE). We show that the approximation error is bounded in its $L_\infty$-norm. We develop a $\emph{Min-Plus}$ Approximate Dynamic Programming (MPADP) algorithm to compute the solution to the MPPBE. We also present the proof of convergence of the MPADP algorithm and apply it to two problems, a grid-world problem in the discrete domain and mountain car in the continuous domain. 
\end{abstract}

\section{Introduction}
Markov Decision Process (MDP) is a useful mathematical framework for posing, analyzing and solving stochastic optimal sequential decision making problems. An MDP is characterized by its state space, action space, the model parameters namely reward structure, and the probability of transition from one state to another under any given action. We consider an MDP with $n$ states and $d$ actions. A policy $u$ specifies the manner in which states are mapped to actions. The value of a state under a policy is the discounted sum of rewards starting in that state and performing actions according to that policy. Thus a given policy $u$ induces a map from the state space to reals. This map is called the value-function, denoted by $J_u \in \R^n$. Solving an MDP means computing the $\emph{optimal}$ value function $J^*=\underset{u}{\max} J_u$ and the $\emph{optimal}$ policy $u^*=\underset{u}{\arg\max} J_u$. The Bellman operator $T$ (\cite{BertB}) is defined using the model parameters of an MDP, and is a map $T\colon \R^n \ra \R^n$. The Bellman Equation (BE) states that $J^*=TJ^*$ (\cite{BertB}), i.e., the optimal value function $J^*$, is a fixed point of $T$. Most methods to solve MDP such as value/policy iteration (\cite{BertB}) are based on solving the BE.\\
\indent The phenomenon called $\emph{Curse of Dimensionality}$ (or simply $\emph{curse}$) refers to the fact that the size of the state space grows exponentially in the number of the state variables. Most problems of practical interest suffer from the $\emph{curse}$, i.e., have large number of states. In such situations it is expensive to compute the optimal policy/value-function and we need to resort to the use of approximate methods. Approximate Dynamic Programming (ADP) refers to an entire spectrum of methods that aim to obtain sub-optimal policies and approximate value-functions. 
Value-function based ADP methods consider a family of functions and pick a function that approximates the value function well. Typically, the family of functions considered is the linear span of a set of basis functions. This is known as linear function approximation (LFA) wherein the value function is approximated as $J^*\approx\tilde{J}=\Phi r^*$. Here $\Phi$ is an $n \times k$ $\emph{feature}$ matrix and $r^* \in \R^k$ is the weight vector with $k<<n$.\\
\indent Given a $\Phi$ matrix, ADP methods vary in the way they learn the weight vector and hence the approximate solution varies across the various ADP methods. In a class of ADP methods (\cite{Tsit}) $r^*$ satisfies the below relation known as the Projected Bellman Equation (PBE). 
\begin{align}\label{pbe}
\Phi r^*= \Pi T \Phi r^*,
\end{align} 
where the projection matrix, $\Pi=\Phi(\Phi^\top D \Phi )^{-1} \Phi^\top$ and $D$ is any positive definite matrix. The approximation error can be bounded as below (\cite{Tsit}):
\begin{align}\label{pbebnd}
||\Phi r^*-J^*|| \propto ||\Pi J^* -J^*||_D.
\end{align}
Alternatively, there are ADP methods such as the Approximate Linear Program (ALP), wherein $r^*$ does not obey a PBE, and is the solution to the below linear program. 
\begin{align}\label{alpb}
 &\mbox{ }\min ~c^\top \Phi r \\
      &\quad s.t\quad \Phi r\geq T\Phi r, \nonumber
\end{align}
where $c \in \R^n$ is such that $c(i)\geq 0, i=1,\ldots,n$ and $\overset{n}{\underset{i=1}{\sum}} c(i)=1$. The approximation error is bounded as below (\cite{ALP}):
\begin{align}\label{alpbnd}
||\Phi r^*-J^*||_{1,c} \propto ||\Pi J^* -J^*||_\infty,
\end{align}
where $||v||_{1,c}=\sum_{i=1}^{n}|v(i)|c(i)$.
It is evident from \eqref{pbebnd} and \eqref{alpbnd} that the choice of ADP method is dictated by the kind of approximation guarantees required in the application at hand.\\
\indent In this paper, we develop a ADP method based on LFA in $\minp$ algebra called $\minp$ approximate dynamic programming (MPADP). The $\minp$ algebra differs from conventional algebra, in that $+$ and $\times$ operators are replaced by $\min$ and $+$ respectively. $\Rm=(\R\cup +\infty,\min,+)$ is a semiring and semimodule $\Rm^n$ can be defined over $\Rm$ in a manner similar to the vector space $\R^n$ over $\R$. Naturally, $J^* \in \Rm^n$, and given an $n\times k$ feature matrix $\Phi$, with columns $\{\phi_j,j=1,\ldots,k\}$ , we consider the set $\mathcal{V}=\{v|v=\Phi\om r\stackrel{\Delta}{=}\min(\phi_1+r(1),\phi_2+r(2),\ldots,\phi_k+r(k), r \in \R^k\}$, where $\om$ in $\Phi \om r$ emphasizes the fact that the approximation is linear in $\minp$. Our function class $\mathcal{V}$ is a subsemimodule as opposed to the subspace in the conventional LFAs. Akin to the PBE \eqref{pbe}, in order to obtain the approximate value function $\tilde{J}=\Phi \om r^*$ we project onto the subsemimodule $\mathcal{V}$, i.e., $r^*$ obeys the following $\minp$ Projected Bellman Equations (MPPBE). 
\begin{align}\label{basic}
\Phi\om r ^*=\Pi_M T\Phi \om r^*, \Phi\om r^* \in \mathcal{V}
\end{align}
where $\Pi_M \colon \R^n \ra \mathcal{V}$, is the $\minp$ projection operator (defined in section~\ref{semiring}).\\
\begin{comment}
Thus $\Pi_M$ closely resembles the $\emph{Fenchel}$ transform (FT), expect for the fact that FT involves $\sup$, and $\Pi_M$ involves $\min$. In this paper, we develop a new ADP method called $\emph{Min-Plus}$ Approximate Dynamic Program (MPADP) algorithm which finds solution to the following equation called $\emph{Min-Plus}$ Projected Bellman Equation (MPPBE),
\end{comment}
\indent Approximate Dynamic Programs based on the $\minp$ semiring have been developed for deterministic control problems \cite{akian,Gaubert} using the fact that the Bellman operator $T$ is $\minp-\emph{linear}$. However, in the case of infinite horizon discounted reward MDP, the presence of probability transition matrix, and discount factor destroys the linearity of the Bellman operator.
This makes our MPADP algorithm significantly different from \cite{akian,Gaubert}. Also the projection operator $\Pi_M$ onto subsemimodules have been studied before in the literature \cite{cohen1996kernels}. Nevertheless, we use them in the context of finding approximate solution to MDPs. Our specific contributions in this paper are as given below.
\begin{enumerate}
\item We develop for the first time an ADP method that makes use of $\minp$ LFA. Another novel aspect of our approach is the $\minp$ PBE.
\item We characterize the approximation error of $\tilde{J}=\Phi \om r^*$, the solution to MPPBE in \eqref{basic}. In particular, we show that the error bound of the form $||J^*-\tilde{J}||_\infty\propto \min_r||J^*-\Phi\om r||_\infty $.
\item We show that $\Pi_M$ is similar to the $\emph{Fenchel}$ transform and the MPPBE equation is similar to the ALP formulation.
\item We present the MPADP algorithm to solve \eqref{basic}. We also provide the proof of convergence for our algorithm.
\item We demonstrate our method on two benchmark planning problems namely the grid world and mountain car.
\end{enumerate}
The rest of the paper is organized as follows. In section~\ref{mdp}, we provide a brief introduction to discounted reward infinite horizon MDPs. In section~\ref{semiring}, we define the $\Rm$ semiring, and semimodules, and the $\minp$ projection operator $\Pi_M$ onto subsemimodules. In section~\ref{fenchel}, we discuss the similarities of the $\minp$ projection operator $\Pi_M$ and the $\emph{Fenchel-Legendre}$ transform. In section~\ref{mpprojected}, we introduce the MPPBE equation and derive the approximation guarantees. Section~\ref{mpadpalgo} contains the MPADP algorithm with a proof of convergence. Section~\ref{exper} contains experiments conducted on the ``grid world" and ``mountain car" problems. In section~\ref{concl}, we present the conclusions and also discuss future work.

\section{Discounted Reward Markov Decision Processes}\label{mdp}
The ADP methods that we develop in this paper are for infinite horizon discounted reward Markov decision processes. Here, we provide a brief overview of MDPs (please refer to \cite{BertB, Puter} for a more detailed presentation). We consider an MDP with state space, $S=\{1,2,\ldots,n\}$ and action set, $A=\{1,2,\ldots,d\}$. We denote by $p_a(i,j)$ the probability of transitioning from state $i$ to $j$ ($i,j \in S$) under action $a \in A$. For simplicity, we assume that all actions $a \in A$ are feasible in every state $s \in S$. The reward is given by the map $g \colon S \rightarrow \mathbf{R}$ and the discount factor is $\alpha$, $0<\alpha<1$.\\
\indent By policy we mean a sequence $\mu=\{\mu_0,\mu_1,\ldots\}$ of functions $\mu_i$ that map states to actions at time $i$. When $\mu_i=\mu, \forall i=1,2,\ldots$, the policy is said to be $\emph{stationary}$. Stationary policies are of two types:
\begin{enumerate}
\item Deterministic, wherein $\mu=\{u,u,\ldots,u,\ldots\}$, where $u \colon S \rightarrow A$. We denote the class of stationary deterministic policies (SDP) by $U$, and a given SDP by $u$.
\item Randomized, wherein $\mu=\{\pi,\pi,\ldots,\pi,\ldots\}$, where given any $s \in S$, $\pi(s,\cdot)$ is a distribution among actions. Thus in state $s$ action $a$ is performed with probability $\pi(s,a)$. We denote the class of stationary randomized policies (SRP) by $\Pi$, and a given SRP by $\pi$.
\end{enumerate}
Under a stationary policy $u$ (or $\pi$) the MDP is a Markov chain and we denote its probability transition kernel by $P_u=(p_{u(i)}(i,j),i=1\mbox{ to }n, j=1\mbox{ to }n)$ (or $P_\pi$). The discounted reward starting from state $s$ following policy $u$ is denoted by $J_u(s)$, where 
\begin{align}\label{disc}
J_u(s)=\mathbf{E}[\sum^\infty_{t=0} \alpha^t g(s_t)|s_o=s,u].
\end{align}
Here $\{s_t\}$ is the trajectory of the Markov chain under $u$. We call $J_u(s)$ the value function for policy $u$.
We denote the optimal policy by $u^*$ where 
\begin{align}\label{optpol}
u^*=\arg\max_{u \in U}  J_u(s), \forall s \in S.
\end{align}
The optimal value function is given by $J^*(s)=J_{u^*}(s), \forall s \in S$. The optimal value function and optimal policy are related by the Bellman equation below:
\begin{align}
\label{bell} J^*(s)&= \max_{a \in A} (g(s)+\alpha\sum^n_{s'=1}p_a(s,s') J^*(s')),\\
\label{bellpolicy} u^*(s)&= \arg\max_{a \in A} (g(s) +\alpha\sum^n_{s'=1}p_a(s,s')J^*(s')).
\end{align}
Once an MDP is posed, our aim is to find $u^*$. Again, once $J^*$ is known, $u^*$ can always be found by plugging $J^*$ in \eqref{bellpolicy}. Thus, in most cases, we are interested in computing $J^*$. 
Taking cue from \eqref{bell} we define the Bellman operator $T \colon \mathbf{R}^n \rightarrow \mathbf{R}^n$ as
\begin{align}\label{bellop}
(TJ)(s)=\max_{a\in A}( g(s)+\alpha \sum^n_{j=1}p_a(s,s')J(s')), J \in \mathbf{R}^n.
\end{align}
Given $J \in \R^n$, $TJ$ is the one-step, greedy value function. Also $J^*$ is a fixed point of $T$ i.e., $J^*=TJ^*$, and from Lemma~\ref{maxnorm}, Corollary~\ref{uni}, it follows that it is also unique (for proofs, please see \cite{BertB}).\\
\begin{lemma}\label{maxnorm}
$T$ is a $\max$-norm contraction operator, i.e., given $J_1, J_2 \in \mathbf{R}^n$
\begin{align}
||TJ_1-TJ_2||_\infty\leq \alpha ||J_1-J_2||_\infty
\end{align}
\end{lemma}
\begin{corollary}\label{uni}
$J^*$ is a unique fixed point of $T$.
\end{corollary}
Further, Bellman operator $T$ exhibits two more important properties presented in the following Lemmas (see \cite{BertB} for proofs)
\begin{lemma}\label{monotone}
$T$ is a monotone map, i.e., given $J_1,J_2 \in \R^n$ such that $J_2\geq J_1$ then $T J_2\geq T J_1$. Further if $J\in \R^n$ is such that $J\geq TJ$, it follows that $J\geq J^*$.
\end{lemma}
\begin{lemma}\label{shift}
Given $J\in \R^n$, and $k \in \R$ and $\mathbf{1} \in \R^n$ a vector with all entries $1$, then 
\begin{align}
T(J+k\mathbf{1})=TJ+\alpha k\mathbf{1}.
\end{align}
\end{lemma}
\begin{comment}
\begin{proof}
Consider the following iteration
\begin{align}\label{val}
J_{n+1}=TJ_n.
\end{align}
Thus $\lim_n ||J_{n+1}-J_n||\ra 0$.
\end{proof}\qed\\
\end{comment}
$J^*$ can also be seen to be the solution to the following linear program
\begin{align}\label{lp}
 &\mbox{ }\min ~c^\top J \\
      &\quad s.t\quad J\geq TJ, \nonumber
\end{align}
where $c \in \R^n, c\geq 0$.\\
Similarly one can define the Bellman operator restricted to a policy $u$ as
\begin{align}\label{belpol}
(T_u J)(s)=g(s)+\alpha \sum_{s'}p_{u(s)}(s,s') J(s'),
\end{align} 
and it is straightforward to show that the value function of policy $u$ obeys the Bellman equation $J_u=T_u J_u$.\\
\indent Due to the $\emph{curse}$, as the number of variables increase, it is hard to compute exact values of $J^*$ and $u^*$. Approximate Dynamic Programming (ADP) methods make use of \eqref{bell} and dimensionality reduction techniques to compute suboptimal policies $\tilde{u}$ instead of $u^*$. ADP methods approximate $J^*$ by means of $\emph{lower}$ dimensional quantities, i.e. $J^*\approx \tilde{J}$, where $\tilde{J} \in V \subset \R^n$. Typically $V$ is the subspace spanned by a set of preselected basis functions $\{\phi_i, i=1,\ldots,k\}, \phi_i \in \R^n$. Let $\Phi$ be the $n\times k$ matrix with columns $\phi_i,i=1,\ldots,k$, and $V=\{\Phi r|r\in \R^k\}$, then the approximate value function $\tilde{J}$ is of the form $\Phi r^*$ for some $r^* \in \R^k$, i.e., $J^*\approx \tilde{J}=\Phi r^*$. Computing $r^* \in \R^k (k<< n)$ is easier than computing $J^* \in \R^n$. Since $J^*$ is not known, one cannot obtain its projection onto $V$. Hence one obtains $r^*$ either as a solution to the PBE in \eqref{pbe} or solution to the ALP \eqref{alpb}. It is important to note that whilst PBE methods are based on value iteration\cite{BertB}, the ALP method is based on the LP formulation \eqref{lp}. Once the approximate value function $\tilde{J}$ is obtained, the suboptimal/$\emph{greedy}$ policy $\tilde{u}$ is obtained as below.
\begin{align}\label{subpol}
\tilde{u}(s)&= \arg\max_{a \in A} (g(s) +\alpha\sum^n_{s'=1}p_a(s,s')\tilde{J}(s')).
\end{align}
The following lemma characterizes the degree of sub-optimality of the greedy policy $\tilde{u}$.
\begin{lemma}\label{subopt}
Let $\tilde{J}=\Phi r^*$ be the approximate value function and $\tilde{u}$ be as in \eqref{subpol}, then 
\begin{align}
||J_{\tilde{u}}-J^*||_\infty \leq \frac{2}{1-\alpha}||J^*-\tilde{J}||_\infty
\end{align}
\end{lemma}
\begin{proof}
We know that
\begin{align}
\label{top} (T_{\tilde{u}}) \tilde{J}(s) 	&= g(s)+\alpha \sum_{s'} p_{\tilde{u}(s)} (s,s') \tilde{J}(s'),\\
\label{bot} J_{\tilde{u}}(s)		&=g(s)+\alpha \sum_{s'} p_{\tilde{u}(s)}(s,s')J_{\tilde{u}}(s').
\end{align}
Hence we can write by subtracting \eqref{top} from \eqref{bot} 
\begin{align}
J_{\tilde{u}}-\tilde{J}&=T_{\tilde{u}}\tilde{J}-\tilde{J} +\alpha P_{\tilde{u}}(J_{\tilde{u}}-\tilde{J})\nn\\
J_{\tilde{u}}-\tilde{J}&=(I-\alpha P_{\tilde{u}})^{-1}(T_{\tilde{u}}\tilde{J}-\tilde{J})\nn\\
||J_{\tilde{u}}-\tilde{J}||_\infty&\leq \frac{1}{1-\alpha}||T_{\tilde{u}}\tilde{J}-\tilde{J}||_\infty.\nn
\end{align}
We know from \eqref{subpol} that $T_{\tilde{u}}\tilde{J}=T\tilde{J}$. Also from the fact that $J^*=TJ^*$ and the contraction property of $T$, we know $||T\tilde{J}-J^*||_\infty\leq \alpha ||\tilde{J}-J^*||_\infty$ and $||T_{\tilde{u}}\tilde{J}-\tilde{J}||_\infty\leq (1+\alpha)||\tilde{J}-J^*||_\infty$. Hence we have
\begin{align}
||J_{\tilde{u}}-J^*||_\infty&=||J_{\tilde{u}}-J^*+\tilde{J}-\tilde{J}||_\infty\nn\\
&\leq  ||J_{\tilde{u}}-\tilde{J}||_\infty+||\tilde{J}-J^*||_\infty\nn\\
&\leq \frac{1}{1-\alpha}||T_{\tilde{u}}\tilde{J}-\tilde{J}||_\infty +||J^*-\tilde{J}||_\infty\nn\\
&\leq \frac{1+\alpha}{1-\alpha}||J^*-\tilde{J}||_\infty +||J^*-\tilde{J}||_\infty\nn\\
&\leq \frac{2}{1-\alpha}||J^*-\tilde{J}||_\infty\nn
\end{align}
\end{proof}
\indent Irrespective of the formulation (PBE or ALP), it is important to choose the basis such that $||J^*-\tilde{J}||_\infty$ is as small as possible. Error bounds for the PBE based methods are in the $L_2$-norm (\cite{Tsit}) and hence the sub-optimality of the greedy policy cannot be ascertained. However, in the case of ALP the sub-optimality of the greedy policy is characterized by error bounds in a modified $L_1$-norm. In this paper, we look at a novel method of approximating $J^*$ using linear function approximators (LFA), which are linear in $\minp$. As we shall see in section~\ref{mpprojected}, our approximate solution has error bounds in the $L_\infty$ norm and hence the sub-optimality of the greedy policy can be ascertained via Lemma~\ref{subopt}. In the next section we describe the $\minp$ LFAs.

\section{Semiring, Semimodules and Projections}\label{semiring}
We define the semiring as $\mathbf{R}_{\min}=(\mathbf{R}\cup \{ +\infty \},\min,+)$ . In $\mathbf{R}_{\min}$, the usual multiplication is replaced with $+$, and addition is replaced by $\min$ given as below.
\begin{definition}
\begin{align}
&\text{Addition:} &x \op y&= \min(x,y)\\
&\text{Multiplication:} &x \om y&= x+y
\end{align}
\end{definition}
Henceforth we use, $(+, \cdot)$ and $(\op,\om)$ to respectively denote the conventional and $\Rm$ addition and multiplication respectively.
In $\mathbf{R}_{\min}$, the multiplicative identity is denoted by $e$ with $e=0 \in \mathbf{R}$ and the additive identity is denoted by $\mathbf{1}$ and is $+\infty$. The $\R_{\min}$ is an idempotent semiring, i.e., $a \op a=a, \forall a \in \Rm$.
We can define a semimodule $\mathcal{M}$ over this semiring, in a similar manner as vector spaces are defined over fields. In particular we are interested in the semimodule $\mathcal{M}=\Rm^n$. Given $u, v \in \Rm^n$, and $\lambda \in \Rm$, we define addition and scalar multiplication as follows:
\begin{definition}
\begin{align}
(u\op v)(i)&=\min\{u(i),v(i)\}=u(i)\op v(i), \forall i=1,2,\ldots,n.\nn\\
(u\om \lambda)(i)&=u(i)\om \lambda=u(i)+\lambda, \forall i=1,2,\ldots,n.\nn\\
\end{align}
\end{definition}
Subsemimodule of semimodules are similar to subspaces of a given vector space. 
\begin{comment}
Given a semimodule $\mathcal{M}$  and a subsemimodule $\mathcal{V}\subset\mathcal{M}$, the $\minp$ $\emph{projection}$-operator $\Pi_M \colon \mathcal{M} \rightarrow \mathcal{V}$ \cite{cohen1996kernels,akian}. Given $u \in \mathcal{M}$, its projection on $\mathcal{V}$ is given by
\end{comment}
The $\minp$ projection operator $\Pi_M$ is given by (\cite{akian,cohen1996kernels,Gaubert})
\begin{align}\label{smproj}
\Pi_M u=\min\{ v | {v \in \mathcal{V}}, v \geq u\}, \forall u \in \mathcal{M}.
\end{align}
In this paper, we consider semimodule $\mathcal{M}=\Rm^n$, and $k$-dimensional subsemimodule $\mathcal{V}$ which is a linear span of a given basis, i.e., $\mathcal{V}= Span\{\phi_i| \phi_i \in \Rm^n, i=1,\ldots,k\}=\{v|v=\Phi \om r\stackrel{\Delta}{=}\phi_1 \om r(1)\op \phi_2 \om r(2)\op \ldots \op \phi_k \om r(k), r(i) \in \Rm, i=1,\ldots,k\}$. We now show that $\Pi_M$ in \eqref{smproj} is closely related to the $\emph{Fenchel}$ transform, or the $\sup$-transform. (For a detailed discussion on projection onto subsemimodules, see \cite{akian}).

\section{Fenchel Dual and Projection on Subsemimodules}\label{fenchel}
In this section, we demonstrate the connections between the $\emph{Fenchel-Legendre}$ transform (FLT) and the $\minp$ projection defined in \eqref{smproj}. Given a function $f \colon \R^n \rightarrow \R$, its FLT is defined by $f^* \colon \R^n \ra \R$, with
\begin{align}\label{FT}
f^*(y)=\sup_{x \in \R^n}(y^\top x-f(x)), y \in \R^n.
\end{align}
If $f$ is convex, then it can be recovered as $f=f^{*^*}$, i.e.,
\begin{align}\label{FFT}
f(x)=f^{*^*}(x)=\sup_{y \in \R^n}(x^\top y-f^*(y)), x \in \R^n.
\end{align}
We can rewrite \eqref{FT} as below
\begin{align}\label{FTRR}
f^*(y)=\sup_{x \in \R^n}(f_y(x)-f(x)), y \in \R^n, \text{ where }f_y(x)=y^\top x.
\end{align}
Now instead of considering functions $f_y(x)$ indexed by $y \in \R^n$, we consider the sequence $\{\phi_j\}, j \in \mathcal{J}=\{1,2,\ldots,k\}, \phi_j \colon \R^n \ra \R$. Then \eqref{FTRR} can be modified as below:
\begin{align}\label{ST}
f^*(j)=\sup_{x \in \R^n}(\phi_j(x)-f(x)), j \in \mathcal{J}.
\end{align}
We call \eqref{ST}, the $\sup$-Transform or the $\max$-Transform. It is easy to check that $\phi_j(x)-f^*(j) <f(x), \forall x \in \R^n, j \in \mathcal{J}$. Since our index set in \eqref{ST} is finite (as opposed to $\R^n$ as in \eqref{FT} ), it is not necessary that the original function $f$ can be $\emph{reconstructed}$ from $f^*(j), j \in \mathcal{J}$. However, we can get an approximation $\tilde{f}$ as below:
\begin{align}\label{AST}
f(x)\approx\tilde{f}(x)=\sup_{j \in \mathcal{J}}(\phi_j(x)-f^*(j)).
\end{align}
In the light of \eqref{ST} and \eqref{AST}, the projection in \eqref{smproj} is nothing but the $\min$-Transform (as opposed to the $\max$-Transform \eqref{ST}). It is more clear if we rewrite \eqref{smproj} for the case when $\mathcal{V}=Span\{\phi_j|\phi_j \in \Rm^n, j=1,\ldots,k)$. Let $\Pi_M u=\Phi \om r^u$, then one can see that
\begin{align}
\Pi_M u&=\{\min\Phi\om r|\Phi \om r \geq u, r \in \Rm^k\}.\\
r^u(j)&= -\min_{i =1,2,\ldots, n} (\phi_j(i)-u(i)), \forall j=1,2,\ldots,k.\label{mpproj}
\end{align}
Note the similarity between $r^u(j)$ in \eqref{mpproj} and $f^*(j)$ in \eqref{ST}.
Then the approximation/projection of $u$ onto $\mathcal{V}$ is given by $\tilde{u}=\Pi_M u=\Phi \om r^u$ with
\begin{align}\label{mplfa}
\Pi_M u(i)&=\min_{j=1,\ldots,k}(\phi_j(i)+r^u(j))\nn\\
&=\phi_1(i) \om r^u(1)\op \ldots \op \phi_k(i) \om r^u(k).
\end{align}
Also, it is important to note that \eqref{ST} deals with projecting a function, while \eqref{smproj} deals with projecting the elements of $n$-dimensional semimodule. Nevertheless, the spirit of the projection is similar in both cases. Also, $\phi_j(i)+r^u_j -u(i)>0$, i.e., the $\min$-Transform approximates the given element $u$ by point-wise minimum of functions that upper bound $u$. We end this section with the following illustration.\\
\begin{example}
Let $f(x)=x^2$, and let $a=(a(j),j=1,\ldots,5)=(-0.8,-0.4,0,0.4,0.8)$, and $\phi_j(x)=2|x-a(j)|$. Then $\minp$ LFA of $f(x)$ via the $\min$-Transform using the $\{\phi_j(x),j=1,\ldots,5\}$ as the $\minp$ basis, is given in the Figure~\ref{minptrans}.
\end{example}
\begin{figure}\label{illust}
\begin{tikzpicture}[scale=1]
    \begin{axis}[
	xlabel=x,
        ylabel=f(x),
	]
    \addplot[smooth,black] plot file {f.dat};
    \addplot[smooth,black] plot file {fproj.dat};
    \addplot[dashed,black] plot file {f1.dat};
    \addplot[dashed,black] plot file {f2.dat};
    \addplot[dashed,black] plot file {f3.dat};
    \addplot[dashed,black] plot file {f4.dat};
    \addplot[dashed,black] plot file {f5.dat};
%    \addplot[dashed,black] plot file {f6.dat};
    \end{axis}
    \end{tikzpicture}
    \caption{$\minp$ LFA of $f(x)$}
    \label{minptrans}
\end{figure}
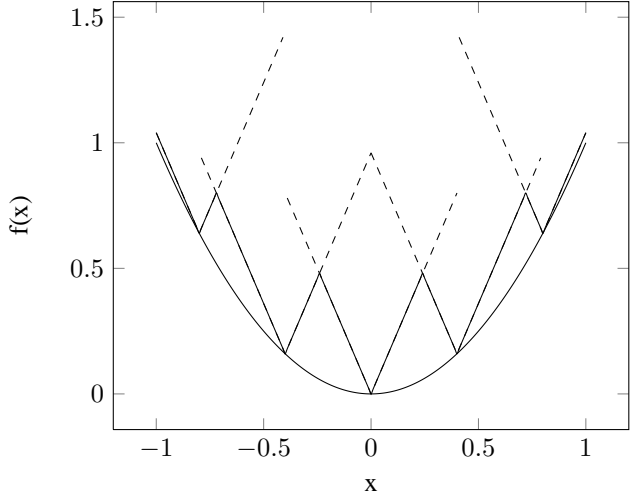

\section{$\minp$ Projected Bellman Equation}\label{mpprojected}
\begin{comment}
Given a $n \times k$ feature matrix $\Phi$, a learnt weight vector $r^* \in \R^k$ the Projected Bellman Equation (PBE) is given by
\begin{align}\label{projbell}
\Phi r^*=\Pi T\Phi r^*.
\end{align}
where $\Pi=\Phi (\Phi^\top \Phi)^{-1}\Phi^\top$ is the projection operator. Given $J \in \R^n$, its projection onto the subspace spanned by the columns of $\Phi$ is $\Phi J$. The aim is to learn $r^*$ such that \eqref{projbell} holds. Here $\Phi r^*$ is the linear function approximation of $J^*$ (\eqref{bell}), and columns of $\Phi$ form the basis.
In a manner similar to the PBE in the conventional linear basis, in this paper, we develop ADP method based on a PBE in the $\minp$ basis. 
\end{comment}
Given a $n\times k$ feature matrix $\Phi$, since we do not know $J^* \in \Rm^n$, $\Pi_M J^*$ cannot be obtained. Thus taking a cue from \eqref{pbe}, we have the approximate value function $\tilde{J}=\Phi\om r^*$ to obey the $\minp$ Projected Bellman Equation (MPPBE) given below:
\begin{align}\label{smprojbellprelim}
\Phi \om r^* =\Pi_M T\Phi \om r^*.
\end{align}
We can expand \eqref{smprojbellprelim} based on \eqref{smproj}, as follows:
\begin{align}\label{smprojbell}
\min\{ \Phi \om r | \Phi \om r  \geq T\Phi \om r, r\in \Rm^k \}.
\end{align}
%Note that the $\minp$ PBE involves the $\min$-Transform described in \eqref{IT} of $T \Phi \om r$.
The above \eqref{smprojbell} is similar to another class of ADP methods called Approximate Linear Program (ALP) in \eqref{alpb}. However, despite the apparent similarity in structure between the ALP \eqref{alpb} and the PBE in the $\minp$ basis \eqref{smprojbell}, the key difference is in the type of basis representation. We assume that \eqref{smprojbell} is feasible, until we establish that fact in Corollary\ref{feas}. We also make the following definition and assumption:
\begin{definition}\label{li}
We call the set of column vectors $\{\phi_i\},i=1,\ldots,k, \phi_i \in \mathbf{R}^n$ of the $n\times k$ matrix $\Phi$ to be linearly independent if $\Phi\om x=\Phi \om y$ $\iff$ $x=y$.
\end{definition}
\begin{assumption}\label{liassump}
The coulmns of the feature matrix $\Phi$ are independent.
\end{assumption}

\begin{lemma}\label{rnew}
Let $r_1, r_2 \in \Rm^k$ be such that $\Phi \om r_1 \geq T \Phi \om r_1$, and $\Phi \om r_2 \geq T \Phi \om r_2$ and let $r_{new}=r_1\op r_2$, then 
\begin{align}
\Phi \om r_{new} \geq T\Phi \om r_{new}\nn
\end{align}
\end{lemma}
\begin{proof}
From Lemma~\ref{monotone}, it follows that
\begin{align}
\Phi \om r_1 &\geq T (\Phi \om r_1 \op  \Phi \om r_2),\label{add1}\\
\Phi \om r_2 &\geq T (\Phi \om r_1 \op  \Phi \om r_2).\label{add2}
\end{align}
From \eqref{add1} and \eqref{add2} we have
\begin{align}
(\Phi \om r_1 )\op (\Phi \om r_2)  &\geq T (\Phi \om r_1 \op  \Phi \om r_2),\label{add3}\\
\Phi \om (r_1 \op r_2) &\geq T \Phi \om (r_1 \op r_2),\label{add4}\\
\Phi \om (r_{new}) &\geq T \Phi \om (r_{new}).\label{add5}
\end{align}
\end{proof}
\subsection{Approximation Guarantees of the $\minp$ PBE}
The minimization in $\minp$ PBE in \eqref{smprojbell} is component-wise. It is desirable to identify an equivalent optimization problem wherein the objective function is not multivalued. To this end, we consider the following program:
\begin{align}\label{eq:mpalp}
 &\mbox{ }\min ~c^\top \Phi\om r \\
      &\quad s.t\quad \Phi\om r\geq T\Phi\om r,\nn\\
	& \text{ where } c^\top \Phi \om r =\sum_{i=1}^n c(i) (\Phi \om r)(i).\nonumber
\end{align}
\begin{lemma}\label{unique}
\eqref{eq:mpalp} has a unique solution.
\end{lemma}
\begin{proof}
Let $r^*_1$ and $r_2^*$ be two distinct solutions of \eqref{eq:mpalp}. Then let $r_{new} = r_1^*\op r_2^*$, and $r_{new}$ is feasible from Lemma~\ref{rnew}. Since $r^*_1$ and $r^*_2$ are distinct, there exists a $j$ such that $r_{new}(j)< r^*_1(j)$ or $r_{new}(j)<r^*_2(j)$, and hence from Assumption~\ref{liassump}, $c^\top \Phi \om r_{new}< c^\top \Phi \om r_1^*=c^\top \Phi \om r_2^*$. This contradicts that fact that $r_1^*$ and $r_2^*$ are optimizers. Thus $r_1^*=r_2^*=r_{new}$.
\end{proof}
\begin{corollary}\label{greater}
Let $r_f$ be any feasible solution and $r^*$ the optimal solution for \eqref{eq:mpalp}. Then $r_f\geq r^*$ ($r_f(i)\geq r^*(i) , i=1,\ldots,k$).
\end{corollary}
\begin{proof}
Let $r_1 \stackrel{\Delta}{=}r_f\op r^*$. From Lemma~\ref{rnew} we know that $r_1$ is feasible, and from Lemma~\ref{unique} that $r_1=r^*$.
\end{proof}
The following Lemma~\ref{univalued}, shows that \eqref{eq:mpalp} and \eqref{smprojbell} are equivalent.
\begin{lemma}\label{univalued}
For any $c \in \R^n$, $c>0$ (all components are positive), program \eqref{eq:mpalp} and the $\minp$ PBE in \eqref{smprojbell} are equivalent.
i.e., $r^* \in \Rm^k$  is a solution to \eqref{smprojbell}  $\iff$  $r^* \in \Rm^k$  is a solution to \eqref{eq:mpalp}.
\end{lemma}
\begin{proof}
Let $r^*_1$ and $r^*_2$ be the solutions to \eqref{smprojbell} and \eqref{eq:mpalp} respectively.\\
$\RA$\\
It clearly follows that $r^*_1$ is feasible for \eqref{eq:mpalp}. Now $r^*_2\leq r^*_1$. Suppose not, then define $r^*_{new}\stackrel{\Delta}{=}r^*_1\op r^*_2$. From Lemma~\ref{rnew} we know that $r_{new}$ is feasible. It then follows that for $c>0$, $c^\top \Phi \om r^*_{new}\leq c^\top \Phi \om r^*_{2}$. But since $c>0$ and $r^*_2$  is the solution to \eqref{eq:mpalp}, which implies $r^*_{new}=r^*_2$, hence $r^*_2\leq r^*_1$.\\
$\LA$\\
It is easy to check that $r^*_2$ is feasible for \eqref{smprojbell}. Then $r^*_1\leq r^*_2$. Suppose not, and let $r^*_{new}\stackrel{\Delta}{=}r^*_1\op r^*_2$. From Lemma~\ref{rnew} we know that $r_{new}$ is feasible. Then we know that $\Phi \om r^*_{new} \leq \Phi \om r^*_1$. But $r^*_1$ is the solution to \eqref{eq:mpalp}, so $r^*_1=r^*_{new}$, hence $r^*_1\leq r^*_2$.
\end{proof}
\begin{lemma}\label{eqprog}
$r^*$ is the optimal solution of \eqref{eq:mpalp} if and only if
 \begin{align}\label{eq:alpmin}
 r^*=&\mbox{ }\arg\min_r ||J^*- \Phi\om r||_\infty \\
      &\quad s.t\quad \Phi\om r\geq T\Phi\om r. \nonumber
\end{align}
\end{lemma}
\begin{proof}
$\Rightarrow$\\
Suppose not. Let $r^*_1$ be the solution to \eqref{eq:mpalp} and $r^*_2$ be the solution to \eqref{eq:alpmin}. Then $\hat{r}=r^*_1\op r^*_2$ is feasible for \eqref{eq:alpmin}. We also know from Lemma~\ref{monotone} that $\Phi \om r^*_2\geq  \Phi \om \hat{r} \geq J^*$,  but we know that $r^*_2$ is solution of \eqref{eq:alpmin}, which implies $r^*_1=r^*_2$.\\
$\Leftarrow$\\
Suppose not. Let $r^*_1$ be the solution to \eqref{eq:mpalp} and $r^*_2$ be the solution to \eqref{eq:alpmin}. Then $\hat{r}=r^*_1\op r^*_2$ is feasible for \eqref{eq:mpalp}. But we from Corollary~\ref{greater} know that $r^*_1\leq \hat{r}$ which is a contradiction. Thus $r^*_1$ and $r^*_2$ must be identical.
\end{proof}
\begin{lemma}\label{appmax}
There exists $\tilde{r} \in \Rm^k$ such that $\Phi \om \tilde{r}\geq T(\Phi \om \tilde{r})$ and $||J^*-\Phi \om \tilde{r}||_\infty \leq \frac{2}{1-\alpha} ||J^*-\Phi \om \bar{r}||_\infty$, where $||V||_\infty =\max_i |V(i)|$, $\bar{r}=\underset{r \in \Rm^k}{\arg\min}||J^*-\Phi \om r||_\infty$. 
\end{lemma}
\begin{proof}
 Let $\epsilon=||J^*-\Phi \om \bar{r}||_\infty$. Now due to the $\max$-norm contraction property of $T$ (Lemma~\ref{maxnorm}), we have $||TJ^*-T\Phi \om \bar{r}|| \leq \alpha \epsilon$. So we know that 
\begin{align}\label{eq:great}
\Phi \om \bar{r} \geq T \Phi \om \bar{r} -(1+\alpha)\epsilon \mathbf{1}. 
\end{align}
Now for any $p \in \mathbf{R}$, let $\tilde{r} =(\bar{r}(1)+p,\bar{r}(2)+p,\ldots, \bar{r}(k)+p)$, then 
\begin{align}\label{eq:tilder}
\begin{split}
\Phi \om \tilde{r}&= \Phi \om \bar{r} + p \mathbf{1}.\\
T \Phi \om \tilde{r}&= T \Phi \om \bar{r} +\alpha p \mathbf{1}.
\end{split}
\end{align}
For $p=\frac{1+\alpha}{1-\alpha} \epsilon$, from \eqref{eq:tilder} and \eqref{eq:great}, we have
\begin{align}
 \Phi \om \tilde{r} - T \Phi \om \tilde{r} &=\Phi \om \bar{r}- T \Phi \om \bar{r} +(1-\alpha)\frac{1+\alpha}{1-\alpha}\epsilon \mathbf{1}\nn\\
&=\Phi \om \bar{r}- T \Phi \om \bar{r} + (1-\alpha)\epsilon \mathbf{1}  \nn\\
&\geq \mathbf{0}. \nn\\
\end{align}
Now 
\begin{align}
 ||J^*-\Phi \om \tilde{r}||_\infty &\leq ||J^*-\Phi \om \bar{r}||_\infty + ||\Phi \om \bar{r} -\Phi \om \tilde{r}||_\infty\nn\\
				   &=(1+\frac{1+\alpha}{1-\alpha})||J^*-\Phi \om \bar{r}||_\infty\nn\\
				   &=\frac{2}{1-\alpha}||J^*-\Phi \om \bar{r}||_\infty.\nn
\end{align}
\end{proof}
\begin{corollary}\label{feas}
\eqref{eq:mpalp} is feasible.
\end{corollary}
We now state the approximation bound
\begin{theorem}
Let $r^*$ be the solution of \eqref{eq:mpalp}, and $\hat{r}=\arg\min_r||J^*-\Phi\om r||_\infty$. Then we have
\begin{align}
||J^*-\Phi\om r^*||_\infty\leq \frac{2}{1-\alpha}||J^*-\Phi \om \hat{r}||_\infty.\nn
\end{align}
\end{theorem}
\begin{proof}
We have shown in Lemma~\ref{appmax} that there exists $\tilde{r}$ feasible such that $||J^*-\Phi \om \tilde{r}||_\infty\leq \frac{2}{1-\alpha}||J^*-\Phi \om \hat{r}||_\infty$. Now we know from Lemma~\ref{eqprog} that $||J^*-\Phi \om r^*||_\infty\leq||J^*-\Phi \om \tilde{r}||_\infty$.
\end{proof}
Thus irrespective of the choice of $c$ the $L_\infty$-norm bound on the approximation error always holds, which is not the case of conventional ALP. Going forward we would want to further understand \eqref{eq:mpalp} and develop an algorithm to solve it. 
\begin{comment}
\begin{definition}\label{part}
Given $r \in \mathbf{R}^k$ we say that the column vector $\phi_j$ $\emph{participates}$ in $V=\Phi \om r$, if $\exists$ an $i$ such that $V(i)=\phi_j(i)+r(j)$
\end{definition}
\end{comment}
\begin{definition}
At a given $r \in \mathbf{R}^k$: 
\begin{enumerate}
\item We say that column vector $\phi_j$ $\emph{participates}$ in row $i$, if $(\Phi\om r)(i)=\phi_j(i)+r(j)$.
\item We call row $i$ to be $\emph{active}$ if $(\Phi \om r)(i)= (T\Phi \om r)(i)$
\end{enumerate}
\end{definition}

\begin{definition}\label{actp}
We call a point $\tilde{r}$ to be an $\emph{active}$-point if the following hold:
\begin{enumerate}
\item Each column of $\Phi$ $\emph{participates}$ in at least one row of $\Phi$. \label{acp1}
\item Atleast one of the rows is active, i.e., $\exists i$ such that $\Phi \om \tilde{r} (i) = (T \Phi \om \tilde{r})(i)$. \label{acp2}
\item Each column of $\Phi$ $\emph{participates}$ in one or more $\emph{active}$ rows. \label{acp3}
\item It is feasible i.e., $\Phi \om \tilde{r}\geq T(\Phi \om \tilde{r})$. \label{acp4}
\end{enumerate}
\end{definition}

\begin{lemma}\label{nopart}
Let $r \in \Rm^k$ be any point feasible point, i.e., $\Phi \om r\geq T(\Phi \om r)$ . Let $g \in \Rm^k$ be defined as $g(j)\stackrel{\Delta}{=}\min_i(\phi_j(i)+r(j)-T(\Phi \om r)(i))$ and $r_{new}$ be defined as $r_{new}\stackrel{\Delta}{=}r-g$. Then $r_{new}$ is feasible.
\end{lemma}
\begin{proof}
Since $r_{new}\leq r$, we have 
\begin{align}
T(\Phi \om r_{new})\leq T(\Phi \om r).\nn
\end{align}
Pick any column $j$, and let $i$ be any row in which column $j$ participates at $r_{new}$. Then we have
\begin{align}
(\Phi \om r_{new})(i)&= \phi_j(i)+r_{new}(j)\nn\\
&=\phi_j(i)+r(j)-g(j)\nn
\end{align}
Now
\begin{align}
&(\Phi \om r_{new})(i) -(T \Phi \om r_{new})(i)\nn\\
&= \phi_j(i)+r(j)-g(j)-(T \Phi \om r_{new})(i)\nn\\
&\geq \phi_j(i)+r(j)-g(j)-(T \Phi \om r)(i)\nn\\
&\geq 0\nn
\end{align} 
\end{proof}
\begin{corollary}\label{nopartcor}
$r_{new}=r-g'$ is feasible for any $g'\leq g$.
\end{corollary}

\begin{comment}
\begin{proof}
Since $r \in \Rm^k$ is feasible, define $V=\Phi \om r -T(\Phi \om r)$. Since column $j$ does not participate in any of the active rows it is easy to check that $d>0$. Let $i$ be any row in which $j$ participates, now we have
\begin{align}\label{constm}
(\Phi \om r_{new})(i)=(\Phi \om r)(i)-d.
\end{align}
\eqref{constm} is true because column $j$ will participate in row $i$ at $r_{new}$ since $r_{new}=r-de_j$. Also due to monotonicity property of $T$, and since $r_{new}< r$ we have
\begin{align}\label{stay}
T(\Phi \om r_{new})\leq T(\Phi \om r)
\end{align}
Now from \eqref{constm} and \eqref{stay} we have
\begin{align}
(\Phi \om r_{new})(i)-T(\Phi \om r_{new})(i)&=(\Phi \om r)(i)-T(\Phi \om r_{new})(i)-d\nn\\
&\geq 	(\Phi \om r)(i)	-T(\Phi \om r)(i)-d\nn\\
&\geq 0				
\end{align}
\end{proof}
\end{comment}

\begin{lemma}\label{neg}
Let $\tilde{r}$ be an active point and $v>0$ be any positive vector in $\mathbf{R}^k$. Then any $r_{new}$ such that $r_{new}\stackrel{\Delta}{=}\tilde{r}-v$ is not feasible.
\end{lemma}
\begin{proof}
Let $j\stackrel{\Delta}{=}\arg\max_{p=1}^k v(p)$. By part \ref{acp3} of Definition~\ref{actp} column $j$ should participate in any one or more $\emph{active}$ rows. So w.l.o.g, we assume that column $j$ participates in the $\emph{active}$ row $i$ at $\tilde{r}$. Then it follows from definition of $j$ that column $j$ participates in row $i$ at $r_{new}$. Now 
\begin{align}
&(\Phi \om r_{new})(i)-(T\Phi \om r_{new})(i)\nn\\
&\leq (\Phi \om \tilde{r})(i)-(T\Phi \om r_{new})(i)- v(j), \label{l1}\\
					&\leq (\Phi \om \tilde{r})(i)-(T\Phi \om \tilde{r})(i)- v(j)+ \alpha v(j),  \label{l2}\\
					&\leq 0.
\end{align}
\eqref{l2} follows from \eqref{l1} from Lemma~\ref{shift}, and due to the fact that $v\leq v(j) \mathbf{1}$, where $\mathbf{1} \in \R^k$ is vector with all entries equal to $1$.
\end{proof}

The following Lemma characterizes the optimal solution of \eqref{eq:mpalp}
\begin{theorem}
$r^*$ is an optimal solution of \eqref{eq:mpalp} $\emph{iff}$ $r^*$ is an $\emph{active}$-point.
\end{theorem}
\begin{proof}\\
$\Rightarrow$\\
Let us assume on the contrary that part \ref{acp1} of Definition~\ref{actp} is not true for $r^*$. Then $\exists$ some $j$ such that $\phi_j$ does not $\emph{participate}$ in any of the rows. Define $d\stackrel{\Delta}{=}\min_i [\phi_j(i)+r^*(j)-(\Phi \om r^*)(i)]$. Now define $r_{new}\stackrel{\Delta}{=}r^*-de_j$ (where $e_j$ is the standard basis with $1$ in the $j^{th}$ coordinate and all other entries set to $0$). From Corollary~\ref{nopartcor} it follows that $r_{new}$ is feasible for \eqref{eq:mpalp} and $r_{new}\leq r^*$, which is a contradiction by Lemma~\ref{unique}. So part \ref{acp1} of Definition~\ref{actp} has to be true for $r^*$.\\
Suppose part \ref{acp2} of Definition~\ref{actp} is not true for $r^*$.
Define $V=\Phi \om r^*-T\Phi \om r^*$. Since $r^*$ is feasible and none of the rows are active we know that $V>0$. Also, none of the columns participate in any of the active rows (since no row is active). Pick any column $j$, and let $d=\min_i (\phi_j(i)+r^*(j)-(T \Phi \om r^*)(i))$, and $r_{new}=r^*-de_j$. Then from Corollary~\ref{nopartcor}, $r_{new}$ is also feasible, but $r_{new} \leq r^*$, which is not possible by Lemma~\ref{unique}. So part \ref{acp2} of Definition~\ref{actp} has to be true for $r^*$.\\
Finally let us assume on the contrary that part \ref{acp3} of Definition~\ref{actp} is not true for $r^*$. Then $\exists$ some $j$ such that $\phi_j$ does not $\emph{participate}$ in any of the $\emph{active}$ rows. Let $\mathcal{I}$ denote the set of $\emph{active}$ rows, and define $d_1\stackrel{\Delta}{=}\min_{i \notin \mathcal{I}}[\phi_j(i)+r(j)- (T\Phi \om r^*)(i)]$, $d_2\stackrel{\Delta}{=}\min_{i \in \mathcal{I}}[\phi_j(i)+r(j)- (T\Phi \om r^*)(i)] $, and $d\stackrel{\Delta}{=}\min\{d_1,d_2\}$. Define $r_{new}\stackrel{\Delta}{=}r^*-de_j$. Now we have 
\begin{enumerate}
\item $i \notin \mathcal{I}$
\begin{align}
&\Phi \om r_{new}(i) - (T\Phi \om r_{new})(i) \nn\\
&\geq (\Phi \om r^*)(i) - (T\Phi \om r_{new})(i) -d\nn\\
&\geq (\Phi \om r^*)(i) - (T\Phi \om r^*)(i) -d\nn\\
&\geq 0\nn
\end{align}
\item $i \in \mathcal{I}$
\begin{align}
&\Phi \om r_{new}(i) - (T\Phi \om r_{new})(i)\nn\\
&= (\Phi \om r^*)(i) - (T\Phi \om r_{new})(i)\nn\\
&\geq (\Phi \om r^*)(i) - (T\Phi \om r^*)(i)\nn\\
&\geq 0\nn
\end{align}
\end{enumerate}
Thus $r_{new}$ is a feasible solution for \eqref{eq:mpalp} and $r_{new}\leq r^*$, which is a contradiction from Lemma~\ref{unique}. So part \ref{acp3} of Definition~\ref{actp} has to be true for $r^*$.
It is easy to check that part~\ref{acp4} holds trivially.\\
$\Leftarrow$\\
Let $\tilde{r}$ be an $\emph{active}$-point. Let the optimal point $r^*$ be different from $\tilde{r}$. We know from part~\ref{acp4} of Definition~\ref{actp} that $\tilde{r}$ is feasible for \eqref{eq:mpalp}. We know from that Corollary\ref{greater} that $\tilde{r}\leq r^*$, which is a contradiction according to Lemma~\ref{neg}. So $\tilde{r}=r^*$.
\end{proof}
\subsection{Finding a feasible point}
We now split the program \eqref{eq:mpalp} in $k$-variables into $k$ programs in one variable each. We call these programs as Sub $\minp$ Projected Bellman Equation (SMPPBE). The $i^{th}$ SMPPBE is given by
\begin{align}\label{eq:salpfull}
 &\mbox{ }\min ~c^\top \phi_i \om r(i)  \\
      &\quad s.t\quad \phi_i\om r(i)\geq T\phi_i \om r(i). \nonumber
\end{align}
The objective in \eqref{eq:salpfull} can be simplified further.
\begin{align}\label{eq:objred}
c^\top \phi_i\om r(i)&=\sum^k_{j=1} c(i) (\phi_i(j)+r(i)) \nonumber \\
&=\sum^k_{j=1} c(i) \phi_i(j)+\sum^k_{j=1} c(i)r(i)\nn\\
&=\sum^k_{j=1} c(i) \phi_i(j)+r(i) \sum^k_{j=1} c(i)
\end{align}
The first term on the right hand side of \eqref{eq:objred} is a constant and since $\sum^k_{j=1}c(i)>0$, the $i^{th}$ SMPPBE can be equivalently written as below:
\begin{align}\label{eq:salp}
 &\mbox{ }\min r(i)  \\
      &\quad s.t\quad \phi_i \om r(i)\geq T\phi_i \om r(i). \nonumber
\end{align}
Let $r_s^*(i)$ be the optimal value of the $i^{th}$ SMPPBE. We define $r_s^* \in \Rm^k$ as $r_s^*=(r_s^*(1),r_s^*(2),\ldots,r_s^*(k))$.
\begin{theorem}
 $r^*_s$ is feasible for \eqref{eq:mpalp}.
\end{theorem}
\begin{proof}
Since $r_s^*(i)$ is the solution for the $i^{th}$ SMPPBE, we know that
 \begin{align}
  \phi_i\om r_s^*(i) &\geq T \phi_i\om r_s^*(i)\label{l1}.\\
 \text{ Hence,} \nn\\
  \phi_i\om r_s^*(i) &\geq T \min\{\phi_1+r_s^*(1),\ldots,\phi_k + r_s^*(k) \}\label{l2},\\
 \text{ or,}\nn\\
  \phi_i\om r_s^*(i) &\geq T \Phi\om r_s^*, \label{l3}
 \end{align}
where \eqref{l2} follows from \eqref{l1} due to the monotonicity property of $T$, and \eqref{l3} follows from \eqref{l2} due to the definition of $\Phi \om r$.
Now since \eqref{l3} is true for every $i$, we have
 \begin{align}
  \min\{\phi_1+r_s^*(1),\ldots,\phi_k + r_s^*(k) \} &\geq T \Phi \om r_s^*, \nn\\
	\text{ or }\nn\\
   \Phi\om r_s^*  &\geq T \Phi \om r_s^*.\nn
  \end{align}
\end{proof}

\section{$\minp$ Approximate Dynamic Programming Algorithm (MPADP)}\label{mpadpalgo}
\begin{algorithm}
\caption{$\minp$ Approximate Dynamic Programming Algorithm}
\begin{algorithmic}[1]
\STATE Start with any feasible point $r_0$, a small number $\epsilon >0$ a small number and $n=0$.
\WHILE{$||g_n||_\infty >\epsilon$} 
\STATE \label{rowfind} Compute the gradient $g_n(j)=\min_{s \in S}(\phi_j(s)+r_n(j)-(T \Phi \om r_n)(s))$.
\STATE $r_{n+1}=r_{n}-g_n$.
\STATE $n=n+1$.
\ENDWHILE
\RETURN $r_{opt}=r_n$, and approximate value function $\tilde{J}=\Phi \om r_{opt}$.
\end{algorithmic}
\label{algo1}
\end{algorithm}
From Lemma~\ref{nopart}, we know that $r_n$ in Algorithm~\ref{algo1} is feasible for all $n$. 
\begin{theorem}
The Algorithm~\ref{algo1} converges in a finite number of iterations for $\epsilon >0$.
\end{theorem}
\begin{proof}
Suppose not, then at each step, the value function decreases by at least $\min_i c(i) \epsilon$. However the objective function is lower bounded. The claim follows.
\end{proof}
It is important to note that when $||g||_\infty=0$, $r_{opt}$ is an $\emph{active}$-point, (Definition~\ref{actp}), i.e., the optimal solution. For any other $\epsilon >0$, $r_{opt}$ is in the $\epsilon-$neighbourhood of the $\emph{active}$ point, as characterized by the following Lemmas.
\begin{lemma}
Let $v \in \R^k$ be any positive vector with $||v||_\infty >\frac{\epsilon}{1-\alpha}$, and $r_{new}$ defined as $r_{new}\stackrel{\Delta}{=}r_{opt}-v$. Then $r_{new}$ is not feasible.
\end{lemma}
\begin{proof}
Let $j=\arg\max^k_{p=1} v(p)$. Now from line~\ref{rowfind} of Algorithm~\ref{algo1}, there is an $i \ni (\Phi \om r_{opt})(i)-(T\Phi \om r_{opt})(i)<\epsilon$. Now
\begin{align}
&(\Phi \om r_{new})(i)-(T\Phi \om r_{new})(i)\nn\\
&\leq (\Phi \om r_{opt})(i)-(T\Phi \om r_{new})(i)-\frac{\epsilon}{1-\alpha}\nn\\ 
& \leq (\Phi \om r_{opt})(i)-(T\Phi \om r_{opt})(i)-\frac{\epsilon}{1-\alpha}+\alpha\frac{\epsilon}{1-\alpha}, \nn\\
&\leq 0.	
\end{align}	
\end{proof}
\begin{corollary}
$r_{opt}-r^*<\frac{\epsilon}{1-\alpha}$, where $r^*$ is the optimal solution to \eqref{eq:mpalp} and $r_{opt}$ is the solution returned by Algorithm~\ref{algo1}.
\end{corollary}
\begin{proof}
We know that $r^*\leq r_{opt}$. Let $v=r_{opt}-r^*$. Now $||v||_\infty<\frac{\epsilon}{1-\alpha}$.
\end{proof}

\section{Experiments}\label{exper}
We test our MPADP algorithm (Algorithm~\ref{algo1}) on a $10 \times 10$ grid world problem. There are a total of $100$ states, i.e., $S=\{1,2,\ldots,100\}$, the co-ordinate $(x_i,y_j)$ is encoded as the state $s=(i-1)\times 10+j$. The reward matrix is as given in Table~\ref{reward}, where each entry is an integer between $1$ and $10$. 
The grid world problem is used to model terrain exploration by autonomous decision making agents (robots).
In each grid position, the agent has $8$ actions corresponding to the $8$ possible directions. In the corners, fewer directions are feasible, and the rest of the directions lead to the current grid position. So $A=\{1,2,\ldots,8\}$. Actions fail with probability of $0.1$ and no movement is made and the same grid position is retained, i.e., $p_a(s,s)=0.1, a \in A, s \in S$, and with probability $0.9$ the agent reaches the intended grid position.
\begin{table}
\begin{tabular}{c|c|c|c|c|c|c|c|c|c|c|c|c}\hline
&$x_{1}$	&$x_{2}$	&$x_{3}$	&$x_{4}$	&$x_{5}$	&$x_{6}$	&$x_{7}$	&$x_{8}$	&$x_{9}$	&$x_{10}$	\\ \hline 
$y_{1}$	&2	&5	&9	&5	&8	&3	&6	&10	&7	&3	\\ \hline 
$y_{2}$	&10	&10	&7	&1	&4	&4	&3	&8	&4	&4	\\ \hline 
$y_{3}$	&1	&2	&4	&10	&3	&9	&8	&5	&9	&5	\\ \hline 
$y_{4}$	&8	&3	&6	&10	&5	&1	&2	&5	&6	&3	\\ \hline 
$y_{5}$	&9	&2	&5	&5	&1	&1	&7	&5	&4	&9	\\ \hline 
$y_{6}$	&9	&2	&1	&5	&2	&2	&2	&4	&10	&2	\\ \hline 
$y_{7}$	&1	&9	&3	&4	&10	&7	&4	&6	&9	&3	\\ \hline 
$y_{8}$	&4	&6	&2	&10	&10	&8	&7	&6	&6	&2	\\ \hline 
$y_{9}$	&3	&6	&2	&4	&6	&7	&8	&9	&7	&3	\\ \hline 
$y_{10}$	&9	&2	&3	&2	&1	&5	&1	&8	&6	&5	\\ \hline 
\end{tabular}
\caption{Grid world with rewards}
\label{reward}
\end{table}
\\
Let $\{\phi_j, j=1,\ldots,k\}, \phi_j \in \Rm^n$ and $\{\phi^i,i=1,\ldots,n\}, \phi^i \in \Rm^k$ be the columns and rows respectively of the feature matrix $\Phi$. Under the feature representation $\Phi$ the similarity of states $s,s'\in S$ is given by the dot product below:
\begin{align}\label{dotp}
<\phi^s,\phi^{s'}>= \phi^s(1)\om\phi^{s'}(1)\op\ldots\op\phi^s(k)\om\phi^{s'}(k).
\end{align}
We desire the following in the feature matrix $\Phi$.
\begin{enumerate}
\item Features $\phi^i$ should have unit norm, i.e., $||\phi^i||=<\phi^i,\phi^i>=\mathbf{0}$, since $\mathbf{0}$ is the multiplicative identity in the $\minp$ algebra.
\item For dissimilar states $s,s'\in S$, we prefer $<\phi^s,\phi^{s'}>=+\infty$, since $+\infty$ is the additive identity in $\minp$ algebra.
\end{enumerate}
Keeping these in mind, we design the feature matrix $\Phi$ for the grid world problem. Since the state space is similar in the connectivity, we aggregate the states based on the reward forming $k$ partitions. Let $g_{\min} = \min_s g(s), s \in S$, $g_{\max}=\max_s g(s), s \in S$ and $L=g_{\max}-g_{\min}$, then we select the features as follows:
\begin{align}\label{feature}
%$$
\phi^{s}(i) = \left\{
        \begin{array}{ll}
            0 & : g(s) \in [g_{\min}+\frac{(i-1)L}{k},g_{\min}+\frac{(i)L}{k}] \\
       1000 & : g(s) \notin [g_{\min}+\frac{(i-1)L}{k},g_{\min}+\frac{(i)L}{k}],
        \end{array}
    \right.\nn\\
%$$
\forall  i=1,\ldots,k.
\end{align}
We use $1000$ in place of $+\infty$, and set $\epsilon=0$ (see Algorithm~\ref{algo1}). It is easy to verify that $\Phi$ in \eqref{feature} has the enumerated properties. The errors are given in Table~\ref{errtable} for discount factors $0.9$ and $0.99$, where $r_{opt}$ is the result returned by the MPADP in  Algorithm~\ref{algo1}, and $\tilde{u}$ is the greedy policy given by
\begin{align}
\tilde{u}=\underset{a \in A}{\arg\max}\bigg( g(s)+\alpha \sum p_a(s,s')\tilde{J}(s')\bigg), \\\text{where} \mbox{ } \tilde{J}=\Phi\om r_{opt}.\nn
\end{align}
The results are plotted in Figure \ref{Vapp}. Note that $\tilde{J}\geq J^*$. Also the errors in the table obey the error bounds. We also noted that the algorithm finds the optimal actions for about $75$ states.
\begin{table}
\begin{tabular}{|c|c|c|}\hline
Error Term & Error for $\alpha=0.9$ & Error for $\alpha=0.99$\\ \hline
$||J^*-\Phi\om r_{opt}||_\infty$ & $9.2768$ & $18.657$\\ \hline
$||J^*-J_{\tilde{u}}||_\infty$ & $9.3248$ & $99.149$\\ \hline
\end{tabular}
\caption{Error Table}
\label{errtable}
\end{table}

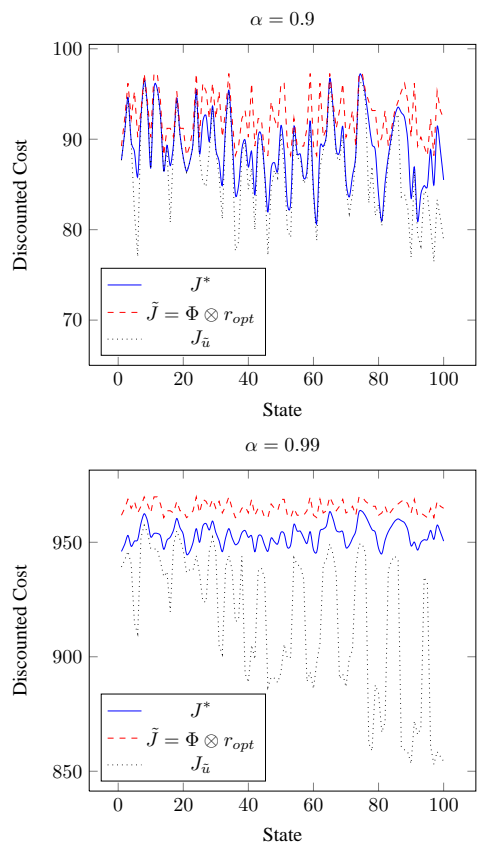
\begin{figure}
\begin{minipage}{0.45\textwidth}
\begin{tabular}{c}
\begin{tikzpicture}[scale=0.75]
    \begin{axis}[
	xlabel=State,
        ylabel=Discounted Cost,
	ymin=65,
	legend pos= south west,
	title={$\alpha=0.9$}
]
    \addplot[smooth,mark=.,blue] plot file {v_p9.dat};
    \addplot[dashed,mark=.,red] plot file {valp_p9.dat};
    \addplot[dotted,mark=.,black] plot file {valppol_p9.dat};
     \legend{$J^*$,$\tilde{J}=\Phi\om r_{opt}$, $J_{\tilde{u}}$}
    \end{axis}
    \end{tikzpicture}
\\

\begin{tikzpicture}[scale=0.75]
    \begin{axis}[
	xlabel=State,
        ylabel=Discounted Cost,
	legend pos=south west,
	title={$\alpha=0.99$}
]
    \addplot[smooth,mark=.,blue] plot file {v_p99.dat};
    \addplot[dashed,mark=.,red] plot file {valp_p99.dat};
    \addplot[dotted,mark=.,black] plot file {valppol_p99.dat};
     \legend{$J^*$,$\tilde{J}=\Phi\om r_{opt}$, $J_{\tilde{u}}$}
    \end{axis}
    \end{tikzpicture}
\end{tabular}
\end{minipage}
\caption{Optimal, Approximate and Greedy Policy Value Function}
\label{Vapp}
\end{figure}
Next we apply the MPADP algorithm to solve the mountain car problem described in the next subsection.
\subsection{Mountain Car}
The problem is to make an underpowered car climb a one-dimensional hill (Figure~\ref{mcar}), whose position $x$ lies in the interval $[-1.2,0.5]$. There are $3$ actions available to the car, i.e., $A=\{0,1,2\}$. $a=0$, $a=2$ correspond to accelerating to left and right respectively. $a=1$ corresponds to no acceleration. The velocity $y$ is limited between $[-0.07,0.07]$. The goal is reached once the car crosses the position $x\geq 0.5$ with a reward of $100$ and everywhere else, the reward is $0$. The dynamics is given by
\begin{align}
y_{t+1}&=y_t+0.001 (a_t-1)-0.0025 cos(3x_t),\\
x_{t+1}&=x_t+y_t.
\end{align}
The state space is continuous with $S=[-1.2,0.5]\times[-0.07,0.07]$ and  the state is given by $s=(x,y), x\in [-1.2,0.5], y \in [-0.07,0.07]$.
The feature vector for state $s$ is
\begin{align}
\phi^s(i)=\big|\beta(\frac{x+1.2}{1.7}-x_i)\big|^\gamma+\big|\beta(\frac{y+0.07}{0.14}-y_i)\big|^\gamma, i=1,\ldots,k,
\end{align}
where $\beta>0$ is a scaling factor and $\gamma>1$ is the order. $(x_i,y_j), i=1,\ldots,k, j=1,\ldots,k$ are the $k \times k$ centers, with $s_{ij}=(x_i,y_j) \in S$. We note that, it is difficult to perform the minimization in line-$3$ of Algorithm~\ref{algo1} over all $s \in S$ and hence we discretize $S$ by means of $k_1 \times k_1$ grid points. These grid points were generated by choosing $x^g_i, i=1,\ldots,k_1$ and $ y^g_j, j=1,\ldots,k_1$, with $s^g_{ij}=(x^g_i,y^g_j)$. 
\\
\begin{comment}
\begin{figure}
\label{mcar}
\input{MOUNTAINCAR.pstex_t}
\caption{Mountain Car}
\end{figure}
\end{comment}
\begin{figure}
\centering
\includegraphics[scale=1.0]{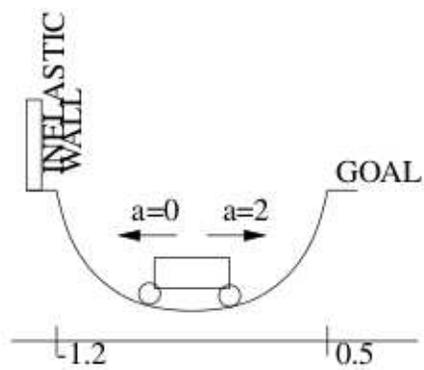}
\caption{Mountain Car}
\label{mcar}
\end{figure}
In our experiments we fixed $\beta=100$ and $\gamma=2$, and varied $k=5, 7, 9, 11$ and $k_1=30,40,50$, and the discount factor was set to $\alpha=0.95$, and $\epsilon=1e^{-5}$. The number of steps taken for the mountain car to reach the goal in each of these settings is presented in Table~\ref{Episode}. The value function learnt in the various cases is presented in Table~\ref{ValFunc}. The actual value function is shown in  Figure~\ref{ActValFunc}. The brighter regions denote higher values and darker regions denote lower values.\\
\begin{table}

\begin{tabular}{|c|c|c|} \hline
$k$ & $k_1$ &Steps to reach the goal\\ \hline	
$5$ &30 &285 \\ \hline
5 &40 &285 \\ \hline
5 &50 &285 \\ \hline

7 &30 &322 \\ \hline
7 &40 &322 \\ \hline
7 &50 &327 \\ \hline

9 &30 &218 \\ \hline
9 &40 &317 \\ \hline
9 &50 &324 \\ \hline

11 &30 &267 \\ \hline
11 &40 &260 \\ \hline
11 &50 &257 \\ \hline
\end{tabular}
\caption{Number of steps taken by the $\emph{Greedy}$ policy}
\label{Episode}
\end{table}\\
Near optimal policy for the mountain car problem is known to achieve the goal within $150$ steps.
\begin{comment}
\begin{table}
\label{ValFunc}
\caption{Approximate Value Function for various values of $k$ and $k_1$}
\centering
\begin{tabular}{cccc}
$k$	& $k_1=30$ & $k_1=40$ &$k_1=50$ \\
$5$
&
\includegraphics[scale=.2]{530.eps}
%\caption{Digraph1.}
%\label{fig:digraph1}
&
\includegraphics[scale=.2]{540.eps}
%\caption{Digraph2.}
%\label{fig:digraph2}
&
\includegraphics[scale=.2]{550.eps}\\
$\mbox{7}$
&
\includegraphics[scale=.2]{730.eps}
%\caption{Digraph1.}
%\label{fig:digraph1}
&
\includegraphics[scale=.2]{740.eps}
%\caption{Digraph2.}
%\label{fig:digraph2}
&
\includegraphics[scale=.2]{750.eps}\\
$9$
&
\includegraphics[scale=.2]{930.eps}
%\caption{Digraph1.}
%\label{fig:digraph1}
&
\includegraphics[scale=.2]{940.eps}
%\caption{Digraph2.}
%\label{fig:digraph2}
&
\includegraphics[scale=.2]{950.eps}\\
$11$
&
\includegraphics[scale=.2]{1130.eps}
%\caption{Digraph1.}
%\label{fig:digraph1}
&
\includegraphics[scale=.2]{1140.eps}
%\caption{Digraph2.}
%\label{fig:digraph2}
&
\includegraphics[scale=.2]{1150.eps}\\
\end{tabular}
\end{table}
\end{comment}
\begin{table}

\centering
\begin{tabular}{ccc}
\tiny
\includegraphics[scale=.2]{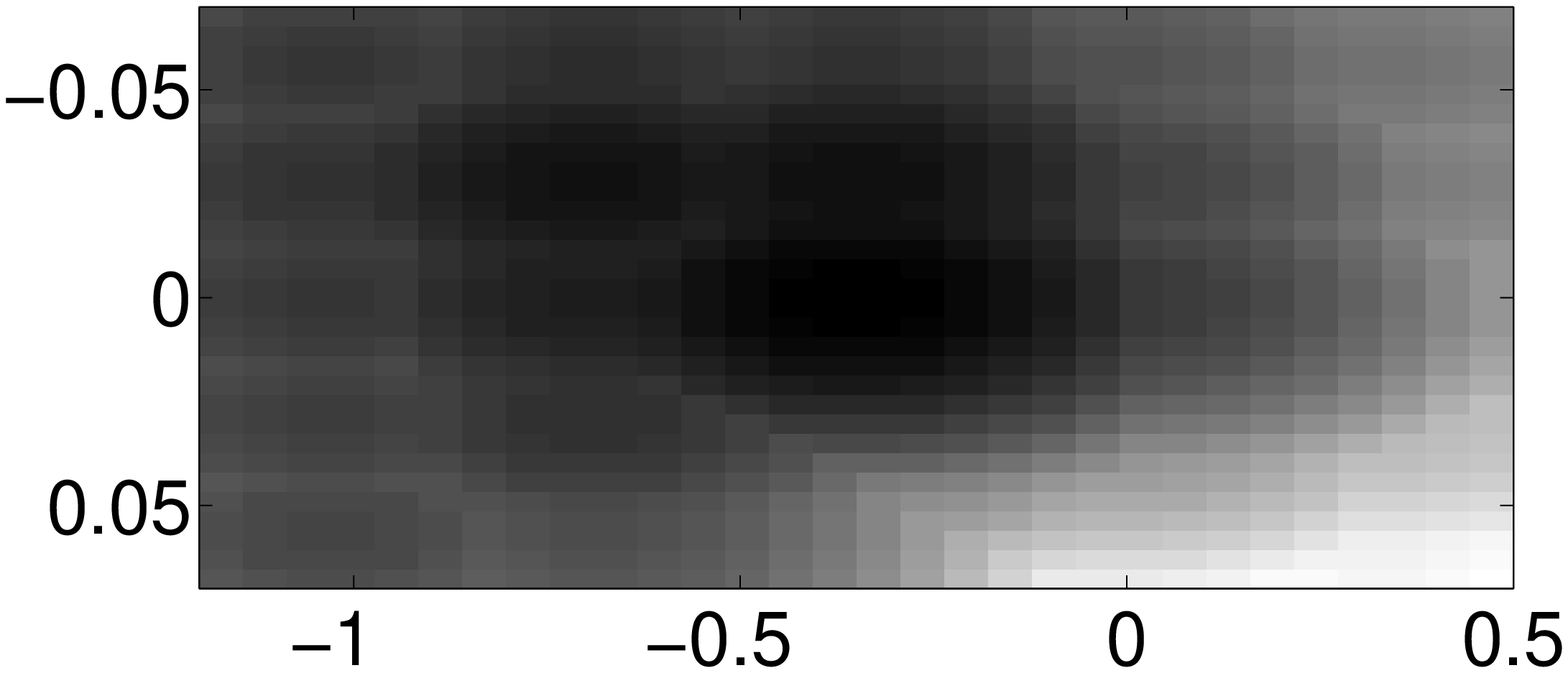}
%\caption{Digraph1.}
%\label{fig:digraph1}
&
\includegraphics[scale=.2]{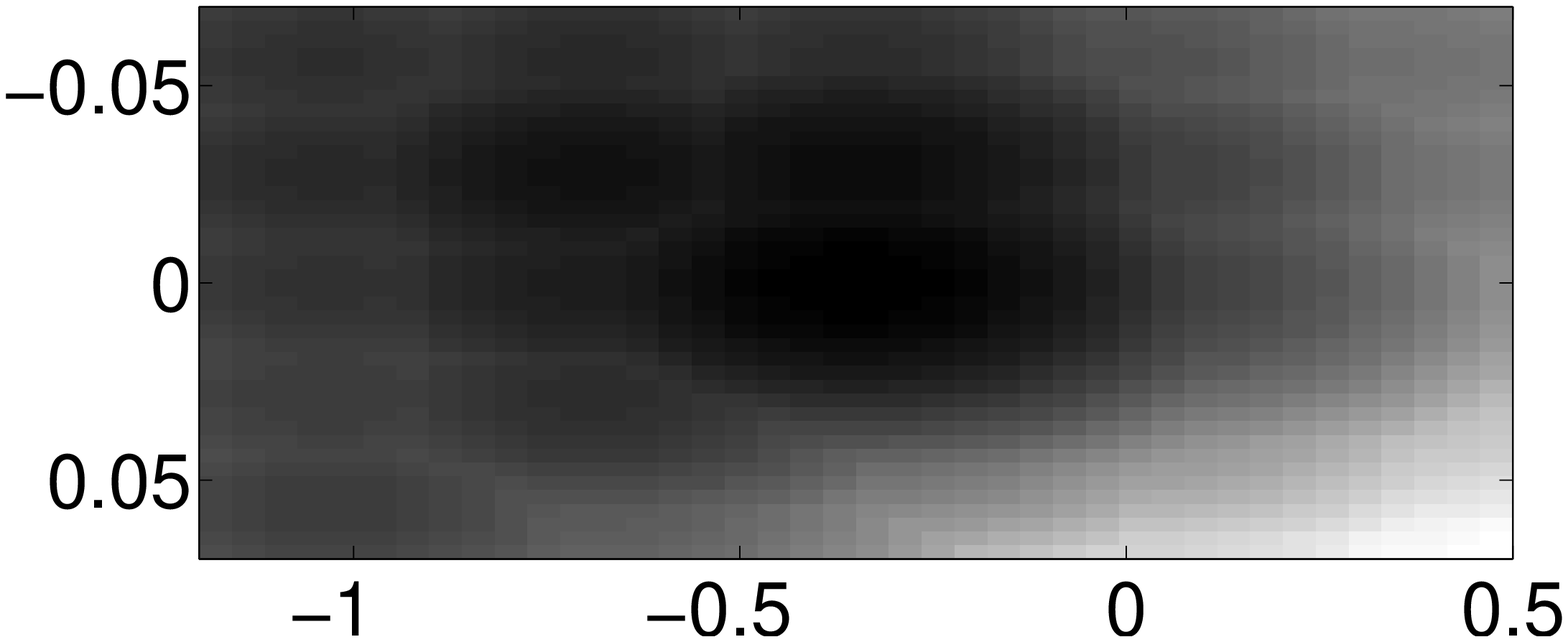}
%\caption{Digraph2.}
%\label{fig:digraph2}
&
\includegraphics[scale=.2]{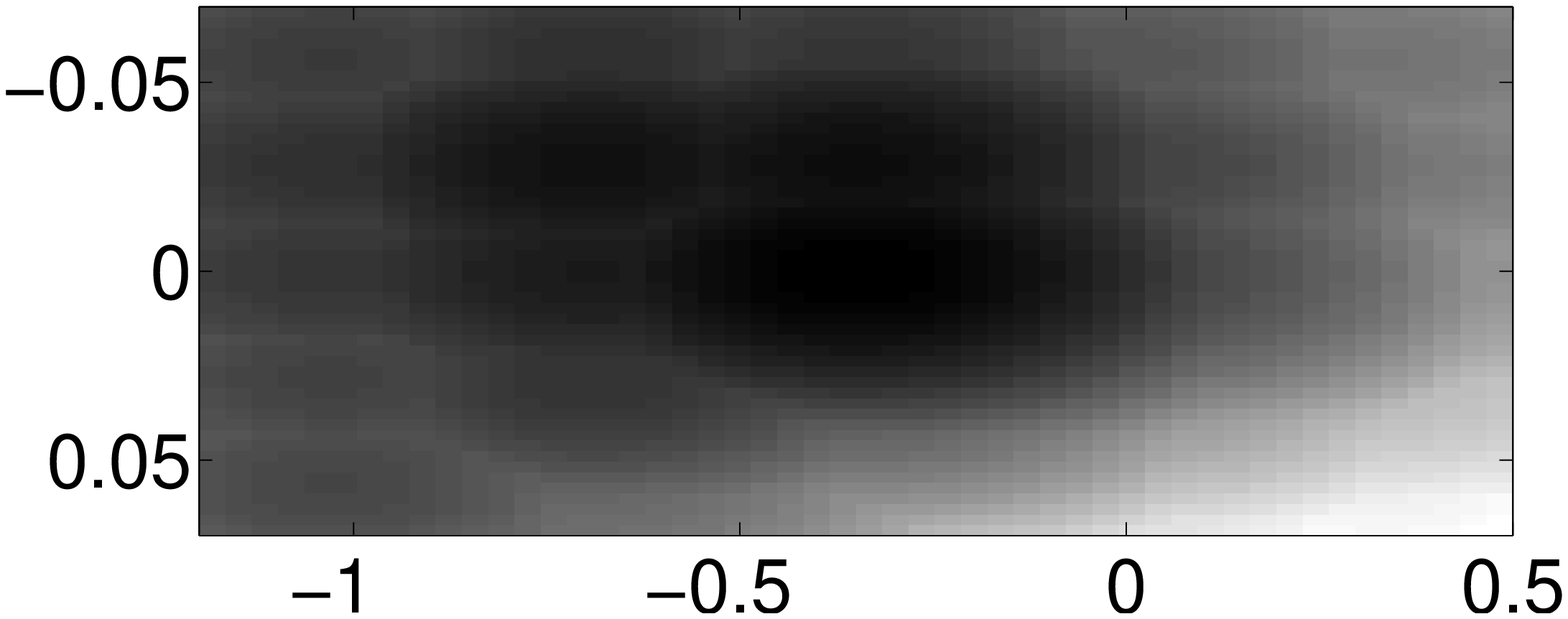}\\
\tiny${k=5,k_1=30,V_{\max}=2.83e3 V_{\min}=0.73e3}$ & \tiny$k=5,k_1=40,V_{\max}=3.29e3, V_{\min}=0.74e3$	&\tiny $k=5,k_1=50,V_{\max}=3.20e3, V_{\min}=0.77e3$\\
\includegraphics[scale=.2]{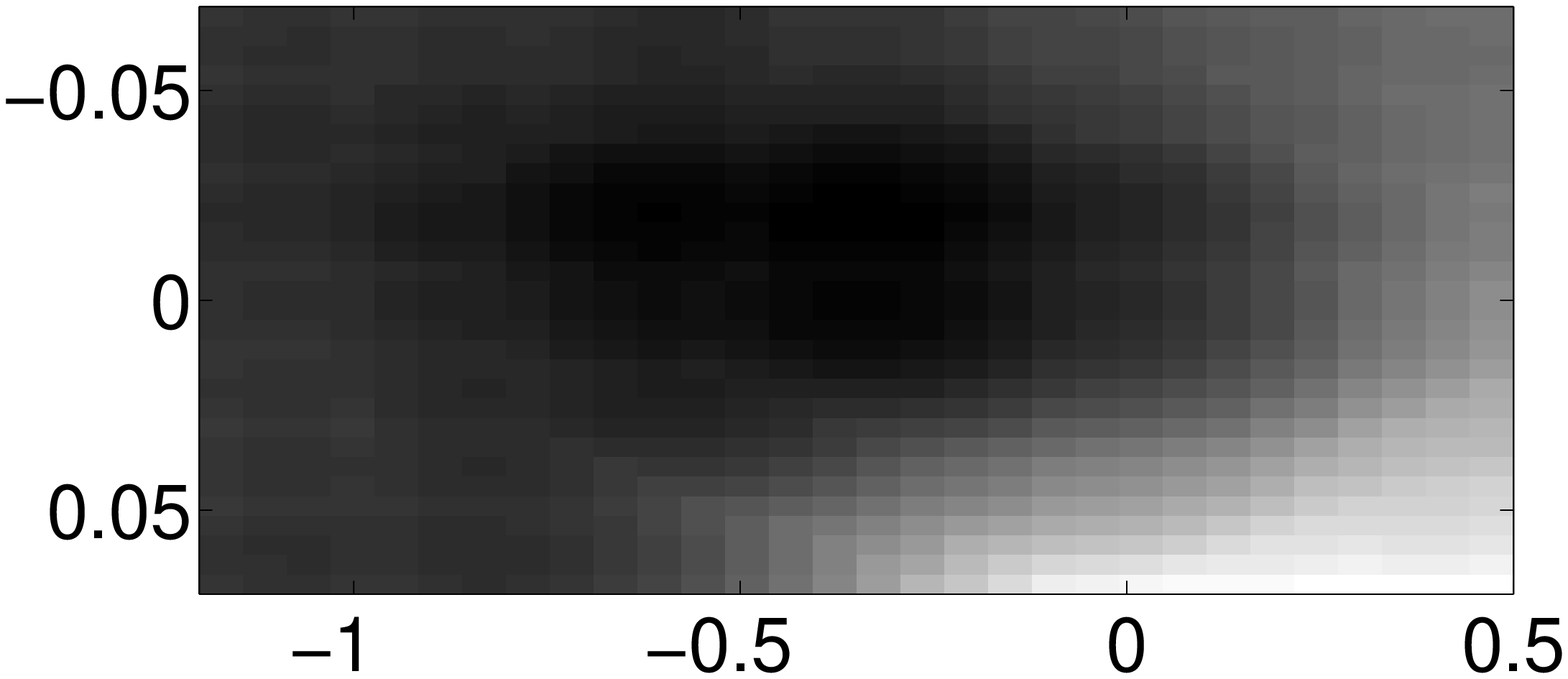}
%\caption{Digraph1.}
%\label{fig:digraph1}
&
\includegraphics[scale=.2]{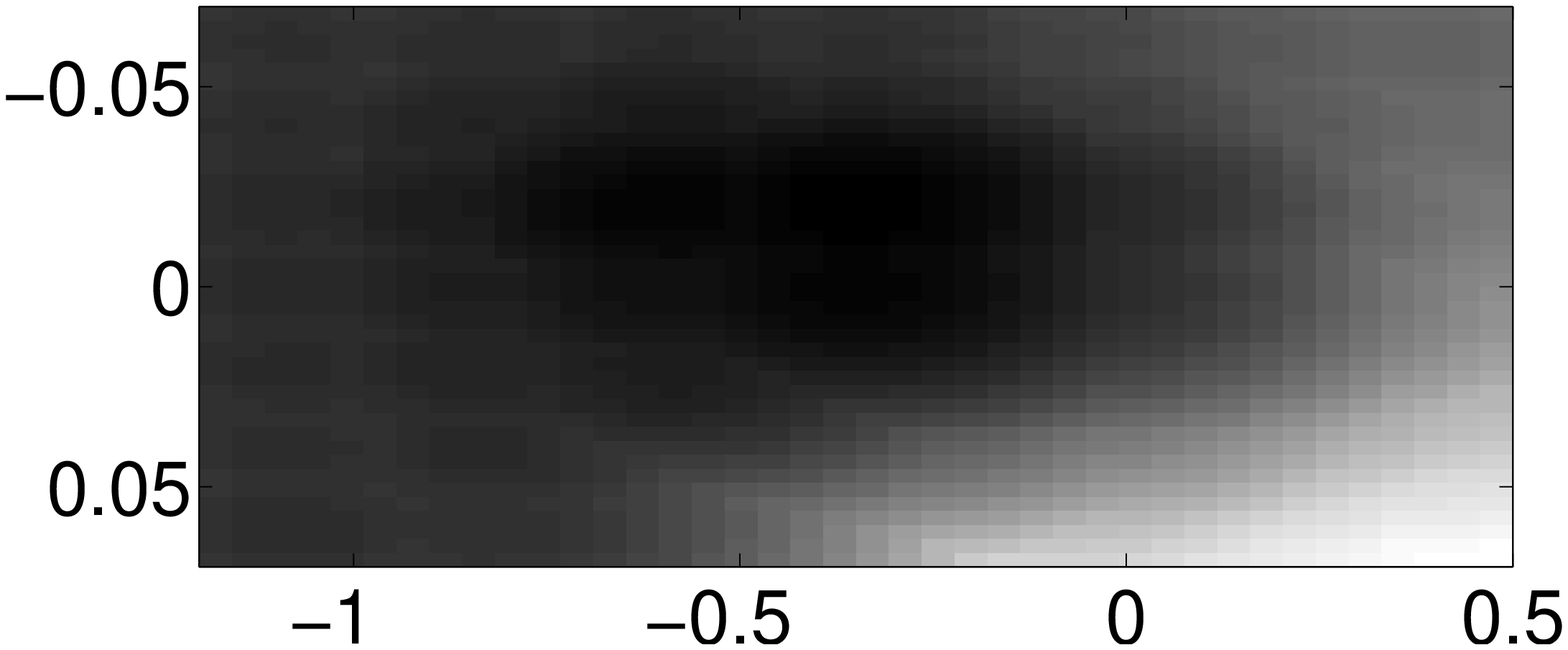}
%\caption{Digraph2.}
%\label{fig:digraph2}
&
\includegraphics[scale=.2]{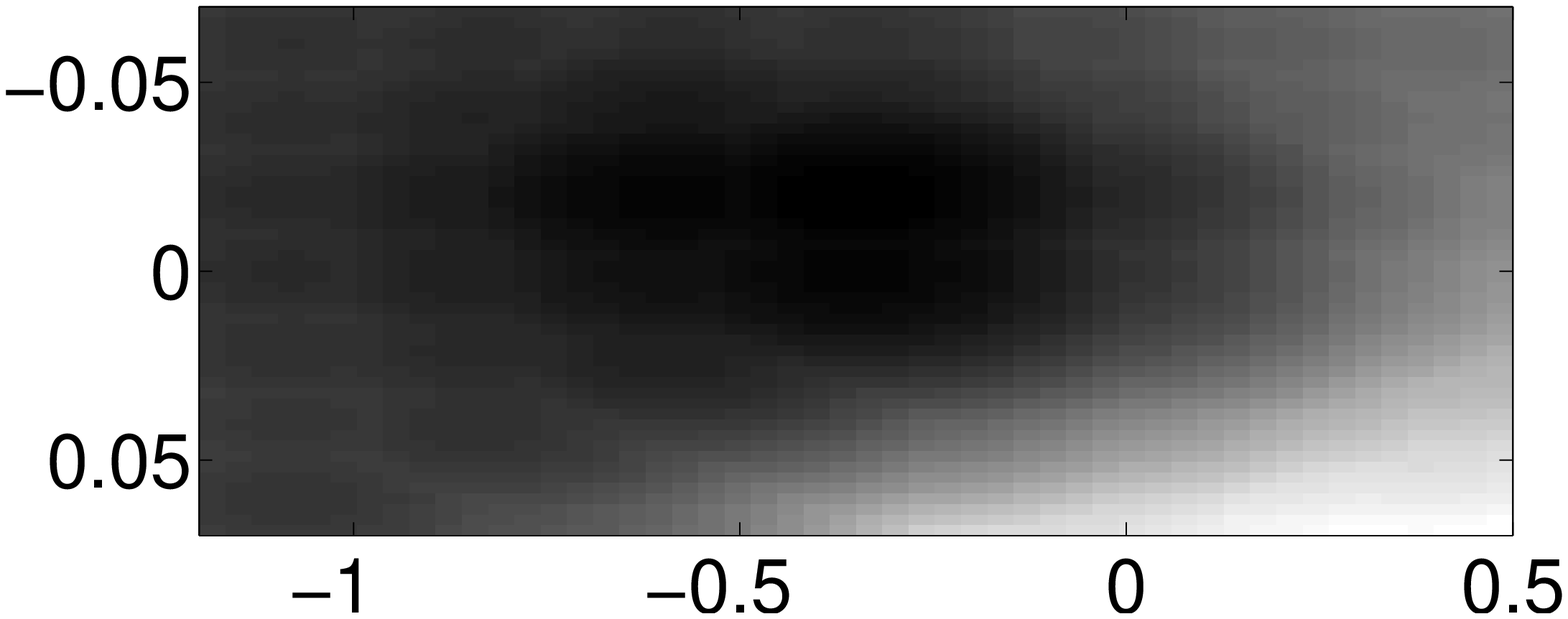}\\
\tiny${k=7,k_1=30,V_{\max}=2.60e3, V_{\min}=0.45e3}$ & \tiny$k=7,k_1=40,V_{\max}=2.83e3, V_{\min}=0.47e3$	&\tiny $k=7,k_1=50,V_{\max}=2.80e3, V_{\min}=0.48e3$\\
\includegraphics[scale=.2]{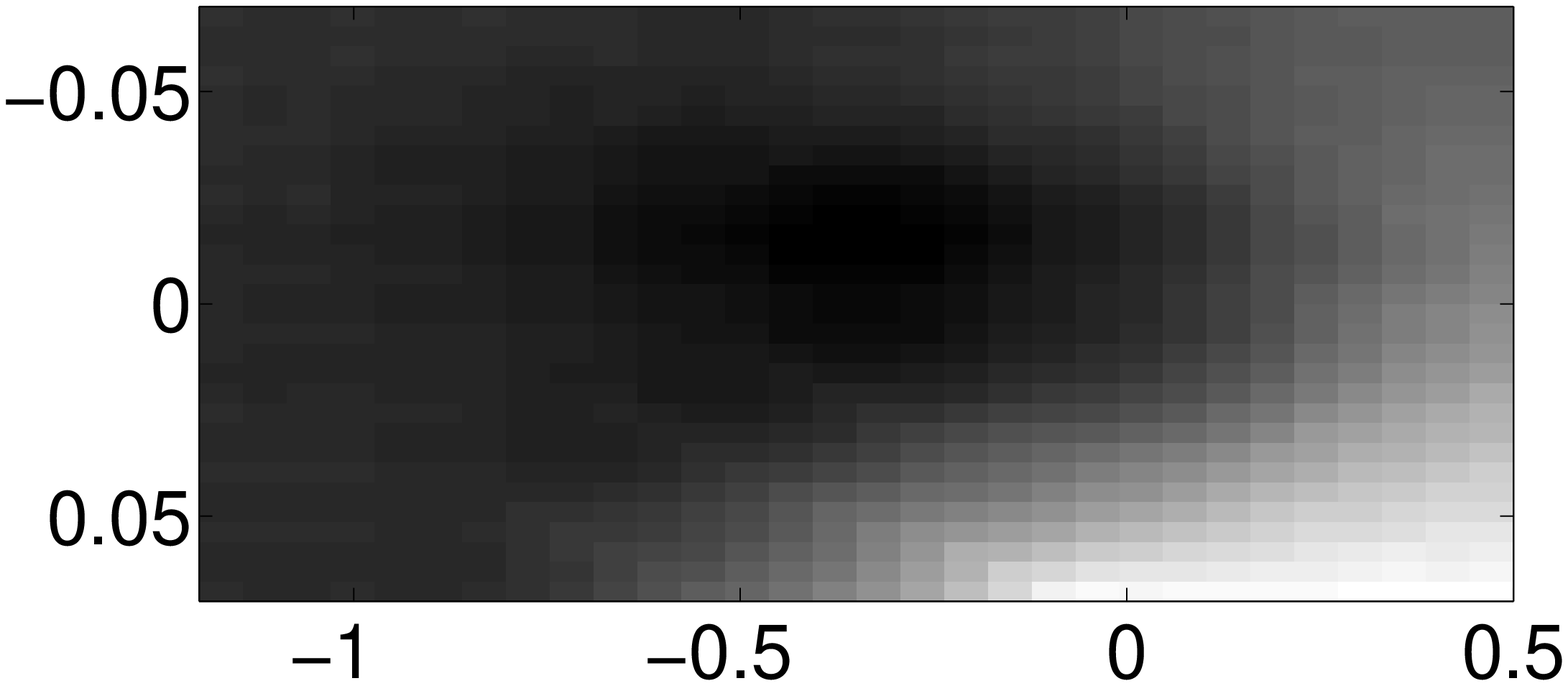}
%\caption{Digraph1.}
%\label{fig:digraph1}
&
\includegraphics[scale=.2]{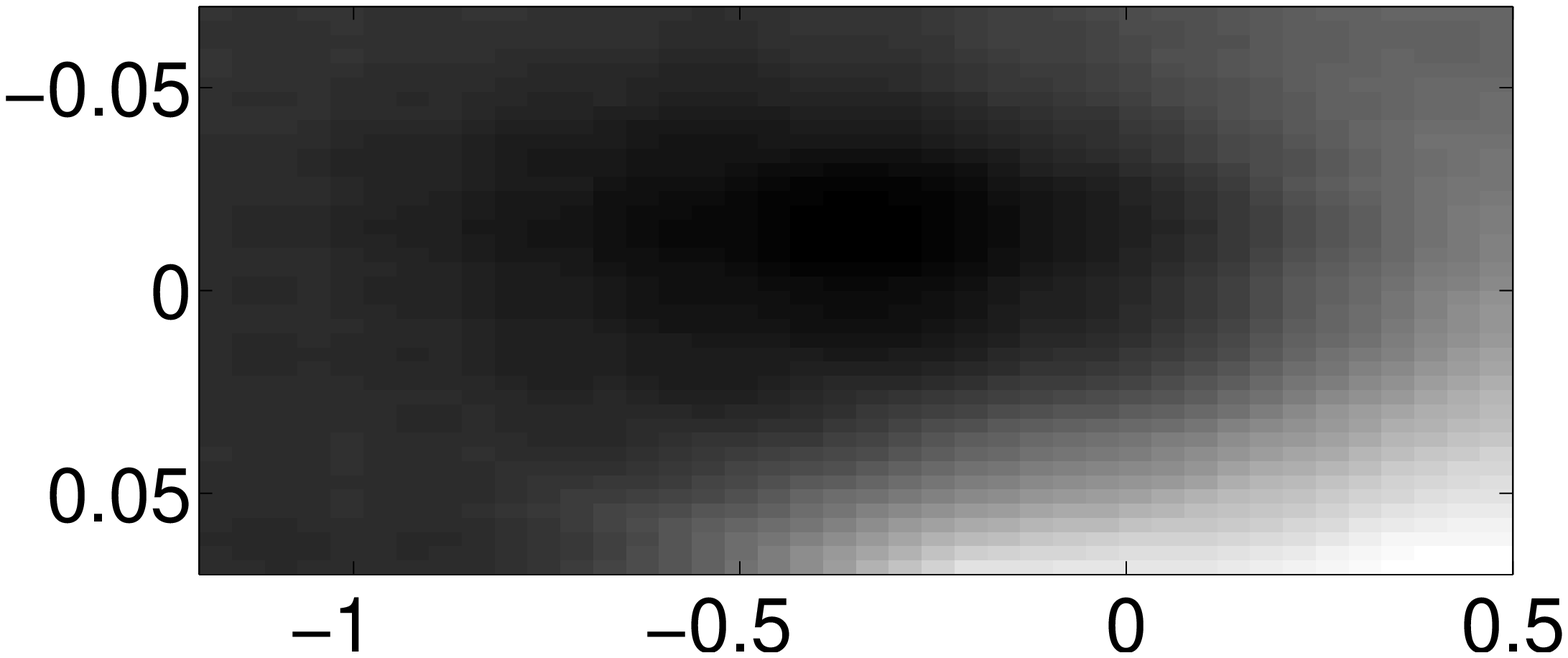}
%\caption{Digraph2.}
%\label{fig:digraph2}
&
\includegraphics[scale=.2]{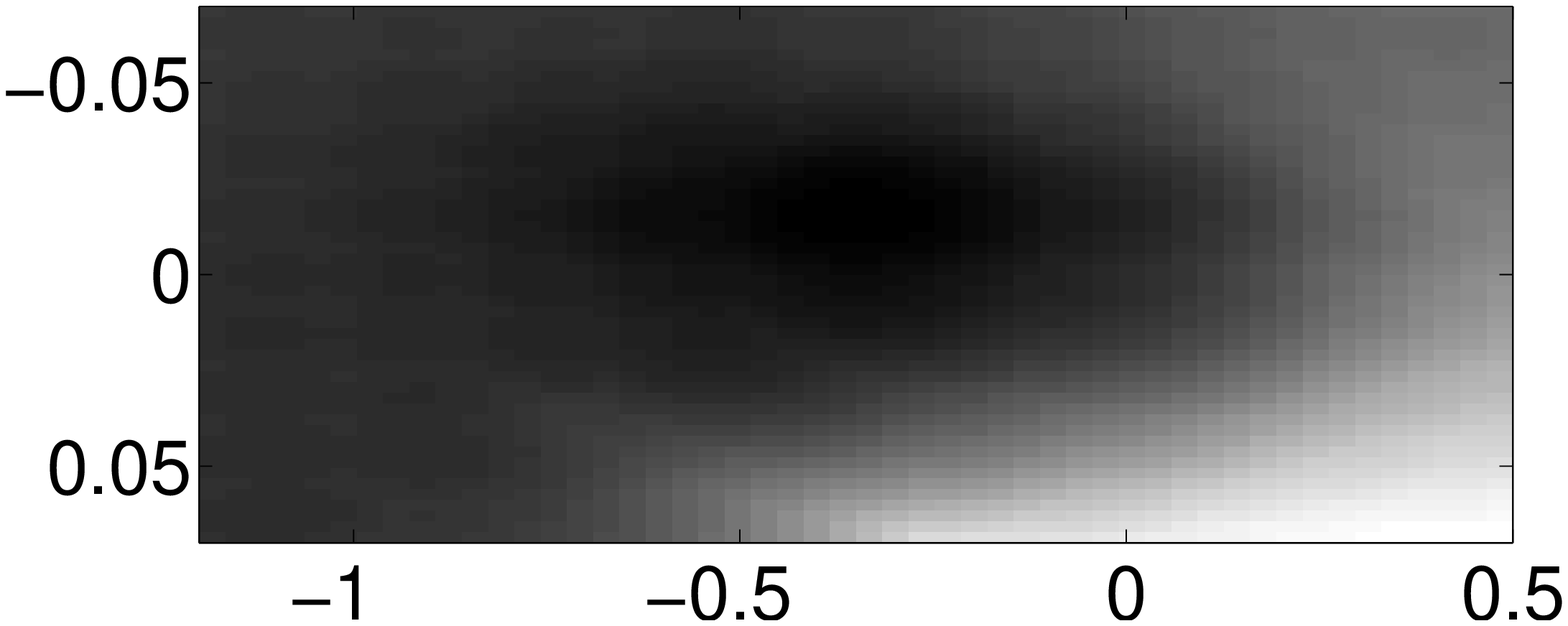}\\
\tiny${k=9,k_1=30,V_{\max}=2.42e3, V_{\min}=0.25e3}$ & \tiny$k=9,k_1=40,V_{\max}=2.58e3, V_{\min}=0.27e3$	&\tiny $k=9,k_1=50,V_{\max}=2.57e3, V_{\min}=0.26e3$\\
\includegraphics[scale=.2]{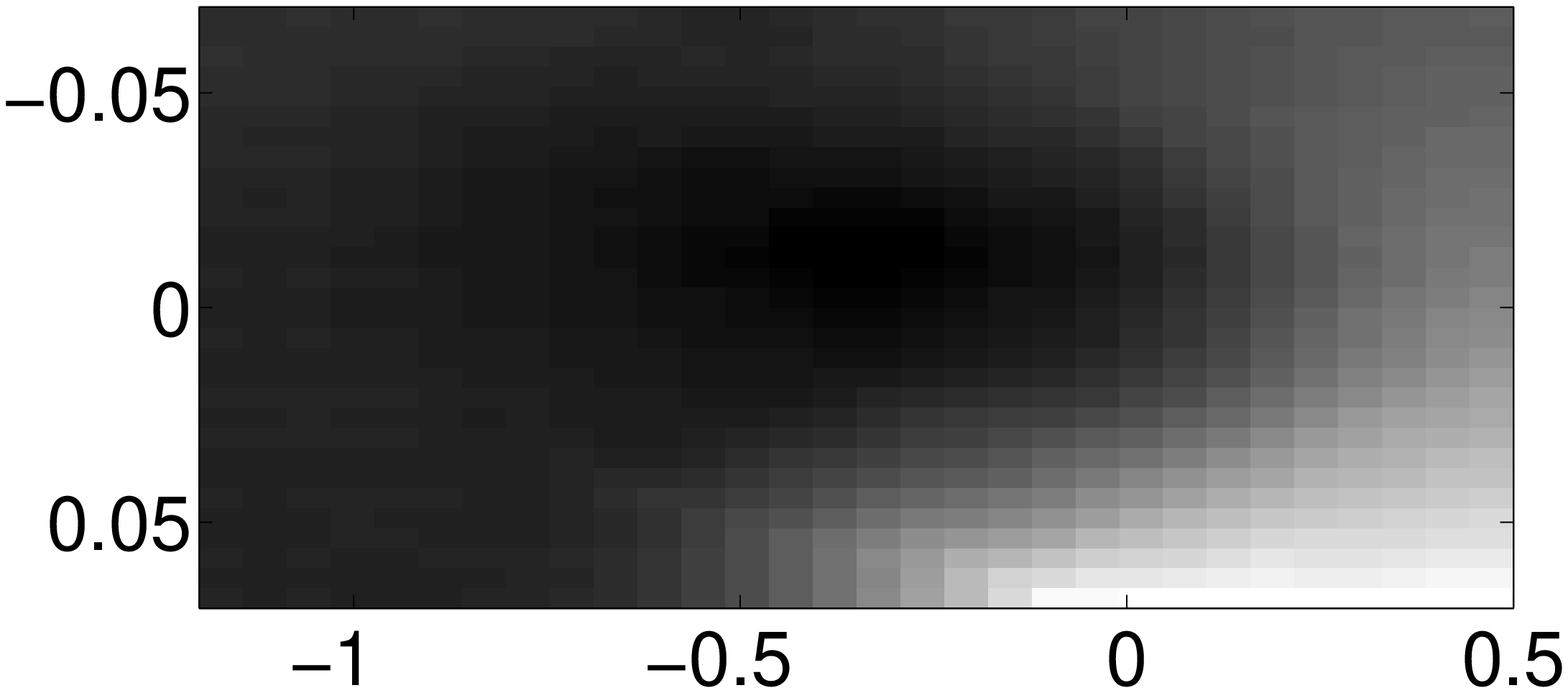}
%\caption{Digraph1.}
%\label{fig:digraph1}
&
\includegraphics[scale=.2]{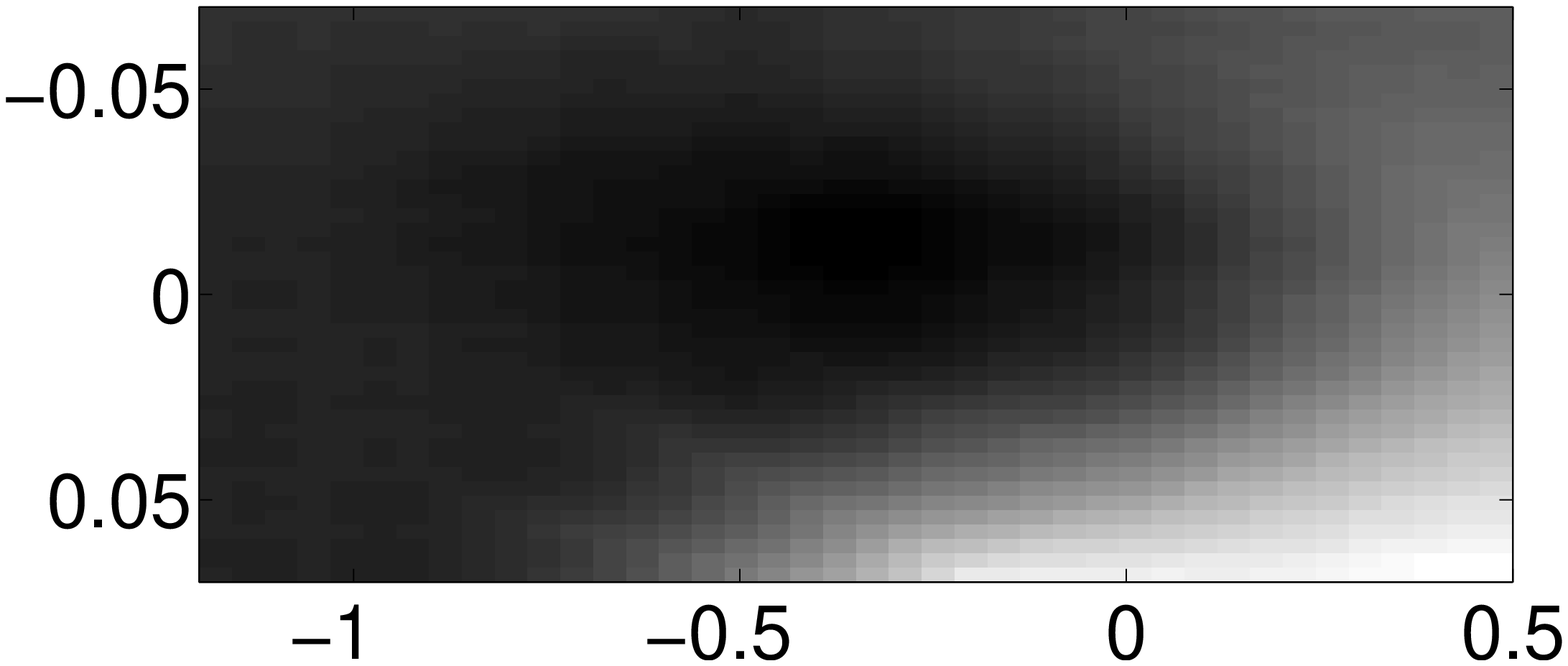}
%\caption{Digraph2.}
%\label{fig:digraph2}
&
\includegraphics[scale=.2]{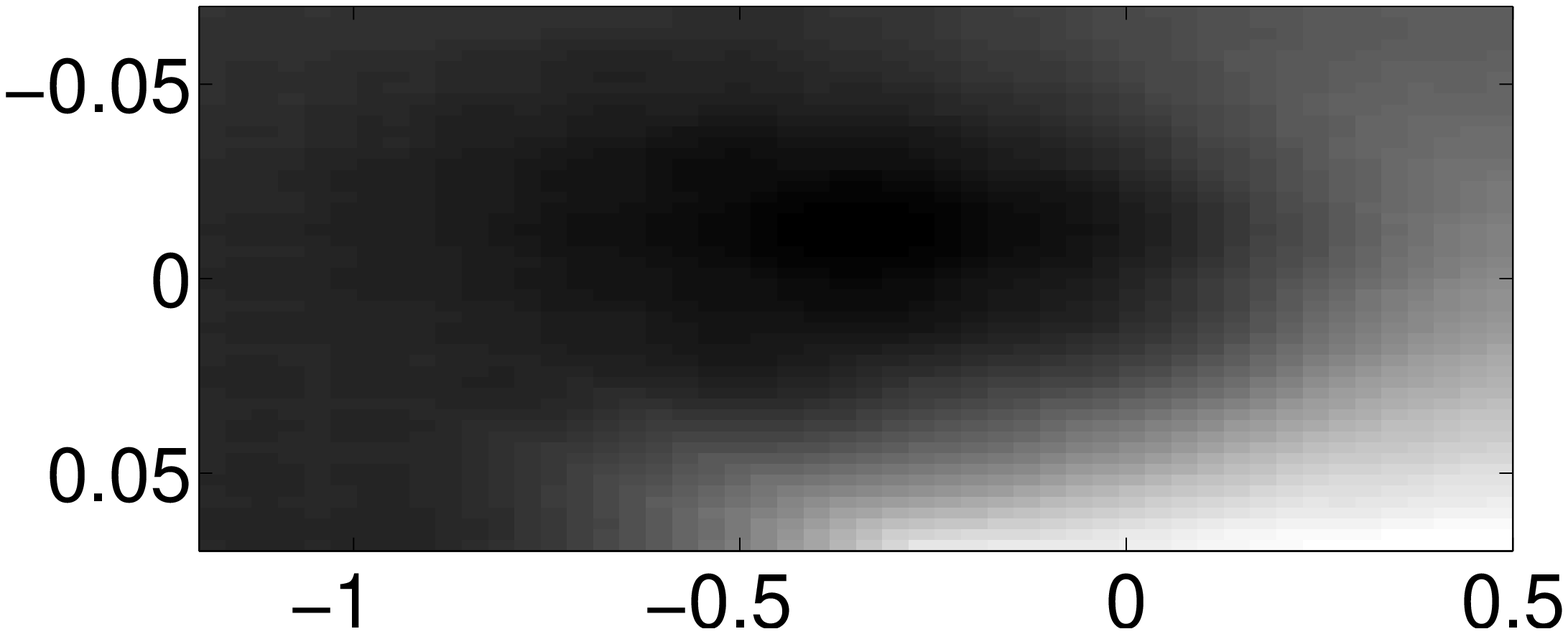}\\
\tiny${k=11,k_1=30,V_{\max}=2.32e3, V_{\min}=0.18e3}$ & \tiny$k=11,k_1=40,V_{\max}=2.42e3, V_{\min}=0.20e3$	&\tiny $k=11,k_1=50,V_{\max}=2.42e3, V_{\min}=0.20e3$
\end{tabular}
\caption{Approximate Value Function for various values of $k$ and $k_1$}
\label{ValFunc}
\end{table}

\begin{figure}
\centering
\begin{tabular}{c}
\tiny
\includegraphics[scale=.3]{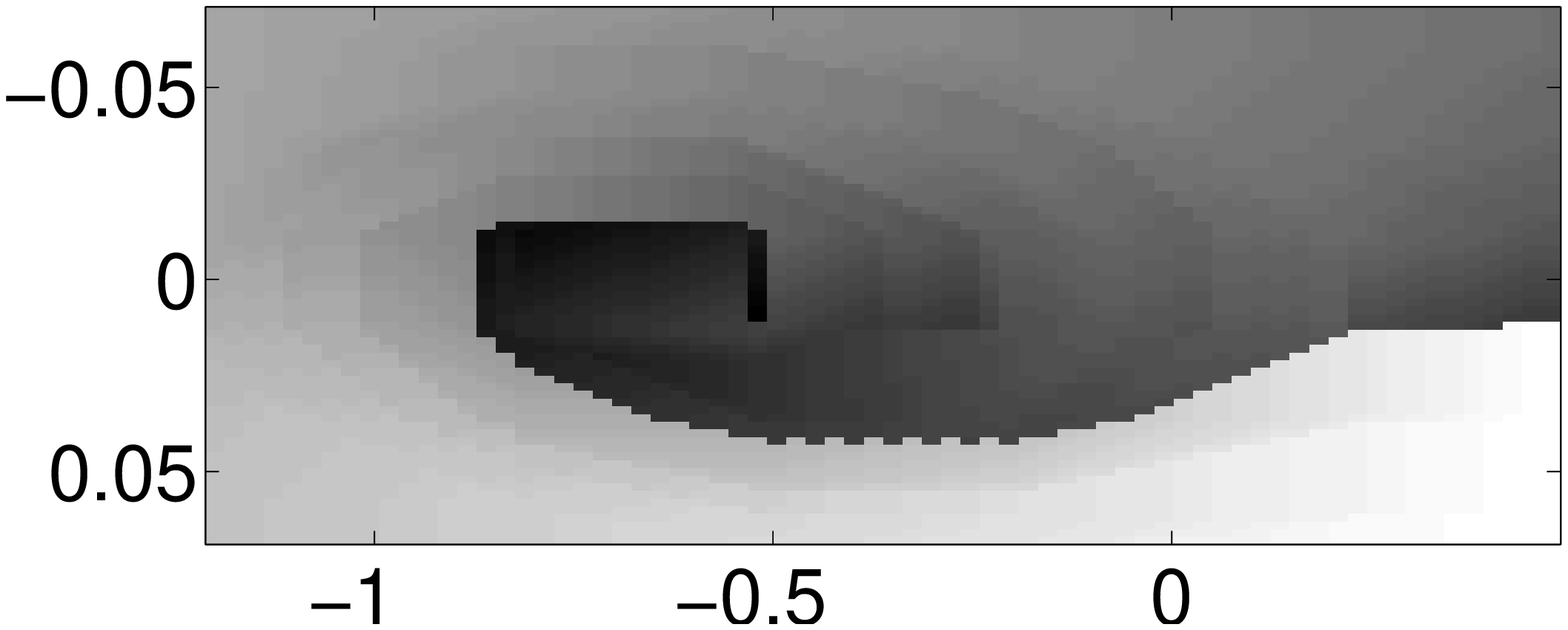}
%\caption{Digraph1.}
%\label{fig:digraph1}
\end{tabular}
\caption{Actual Value function}
\label{ActValFunc}
\end{figure}

\section{Conclusion}\label{concl}
We introduced a novel ADP method to approximate the value function of infinite horizon discounted reward MDP. The novelty was in the use of $\minp$ linear basis as opposed to the conventional linear basis. Our approximate value function belonged to the subsemimodule formed by the $\minp$ linear span of the basis and obeyed the $\minp$ Projected Bellman Equation (MPPBE). The salient feature of the approximate value function was that the error was bounded in the $L_\infty$ norm. We also presented the MPADP algorithm (Algorithm~\ref{algo1}) to solve the MPPBE and showed that the algorithm converges to the desired solution. We also applied our method on two example problems.\\
\indent The use of $\minp$ LFAs in ADP methods is quite new and there are several interesting directions that can be furthered. A question of immediate interest is to find the possibilities of a reinforcement learning (RL) algorithm based on $\minp$ LFA, that solve MDP in the absence of model information. It will be interesting to investigate whether it is possible to develop $Q$-learning algorithm using $\minp$ LFA. Also, further research is required to find the right choice of basis functions in the new algebra. These might together throw light on the right kind of LFA architecture to be chosen for any given problem.

\clearpage
\bibliographystyle{plain}
\bibliography{ref}
\end{document}